\documentclass[format=acmsmall, review=false]{acmart}
\usepackage{acm-ec-26}
\usepackage{booktabs} %
\usepackage[ruled]{algorithm2e} %

\SetAlFnt{\small}
\SetAlCapFnt{\small}
\SetAlCapNameFnt{\small}
\SetAlCapHSkip{0pt}
\IncMargin{-\parindent}

\usepackage{threeparttable}
\usepackage{tikz}
\usetikzlibrary{shapes}
\usetikzlibrary{calc}
\usetikzlibrary{decorations.pathreplacing}

\setcitestyle{authoryear}
\allowdisplaybreaks

\title[Deep DSIC Duality]{Duality for Optimal Multi-Item, Multi-Bidder Auction Design: Revenue Certificates through Deep Learning}

\author{Yanchen Jiang}
\email{}                       %
\affiliation{%
  \institution{Harvard University}
  \city{Cambridge} \state{MA} \postcode{02138} \country{USA}
}

\author{David C. Parkes}
\email{}                       %
\affiliation{%
  \institution{Harvard University}
  \city{Cambridge} \state{MA} \postcode{02138} \country{USA}
}

\author{Tonghan Wang}
\authornote{Corresponding author.}
\email{twang1@g.harvard.edu}
\affiliation{%
  \institution{College of AI, Tsinghua University}
  \city{Beijing} \country{China}
}

\begin{abstract}

\end{abstract}

\usepackage{amsmath}
\usepackage{mathtools}
\usepackage{amsthm}
\usepackage{wrapfig}

\usepackage[capitalize,noabbrev]{cleveref}

\usepackage{thmtools}
\makeatletter
\def\thmheadbrackets#1#2#3{%
  \thmname{#1}\thmnumber{\@ifnotempty{#1}{ }\@upn{#2}}%
  \thmnote{ {\the\thm@notefont[#3]}}}
\makeatother

\newtheoremstyle{brakets}%
  {}%
  {}%
  {\itshape}%
  {}%
  {\bfseries}%
  {.}%
  { }%
  {\thmheadbrackets{#1}{#2}{#3}}%

\theoremstyle{brakets}

\newtheorem{theorem}{Theorem}

\newtheorem{lemma}[theorem]{Lemma}

\newtheorem{definition}[theorem]{Definition}

\theoremstyle{remark}

\definecolor{darkgreen}{rgb}{0.0, 0.5, 0.0}
\definecolor{darkblue}{rgb}{0.0, 0.5, 1.0}

\newcount\Comments  %
\newcommand{\kibitz}[2]{\ifnum\Comments=1{\color{#1}{#2}}\fi}

\newcount\CommentsAdd  %
\newcommand{\kibitzAdd}[2]{\ifnum\CommentsAdd=1{\color{#1}{#2}}\fi}
\definecolor{english}{rgb}{0.0, 0.5, 0.0}
\definecolor{tw}{rgb}{0.0, 0.0, 0.5}

\usepackage{xcolor}

\usepackage{amsmath,amsfonts,bm}

\def\eqref#1{equation~\ref{#1}}

\def\vw{{\bm{w}}}

\DeclareMathAlphabet{\mathsfit}{\encodingdefault}{\sfdefault}{m}{sl}
\SetMathAlphabet{\mathsfit}{bold}{\encodingdefault}{\sfdefault}{bx}{n}

\newcommand{\supp}{\mathrm{supp}}
\newcommand{\Rev}{\mathrm{Rev}}
\newcommand{\E}{\mathbb{E}}

\newcommand{\dV}{\,\mathrm{d}v}

\usepackage{multirow}
\usepackage{makecell}
\newcolumntype{L}{>{$}l<{$}}
\newcolumntype{C}{>{$}c<{$}}
\newcolumntype{R}{>{$}r<{$}}

\crefname{section}{Section}{Sections}
\Crefname{section}{Section}{Sections}
\crefname{chapter}{Chapter}{Chapters}
\Crefname{chapter}{Chapter}{Chapters}
\crefname{appendix}{Appendix}{Appendices}
\Crefname{appendix}{Appendix}{Appendices}
\crefname{equation}{Eq.}{Eqs.}
\Crefname{equation}{Eq.}{Eqs.}
\crefname{part}{Part}{Parts}
\Crefname{part}{Part}{Parts}
\crefname{table}{Table}{Tables}
\Crefname{table}{Table}{Tables}
\crefname{figure}{Figure}{Figures}
\Crefname{figure}{Figure}{Figures}
\crefname{theorem}{Theorem}{Theorems}
\Crefname{theorem}{Theorem}{Theorems}
\crefname{proposition}{Proposition}{Propositions}
\Crefname{proposition}{Proposition}{Propositions}
\crefname{corollary}{Corollary}{Corollaries}
\Crefname{corollary}{Corollary}{Corollaries}
\crefname{lemma}{Lemma}{Lemmas}
\Crefname{lemma}{Lemma}{Lemmas}
\crefname{finding}{Finding}{Findings}
\Crefname{finding}{Finding}{Findings}
\crefname{assumption}{Assumption}{Assumptions}
\Crefname{assumption}{Assumption}{Assumptions}
\crefname{algorithm}{Algorithm}{Algorithms}
\Crefname{algorithm}{Algorithm}{Algorithms}

\begin{document}

\begin{abstract}
    Characterizing revenue-optimal auctions for multi-item, multi-bidder settings remains a fundamental open problem, with no known closed-form solution existing beyond restrictive binary-type instances. This has motivated interest in computational approaches to optimal auction design. In this paper, we introduce the first computational framework that directly tackles the dual problem for multi-item, multi-bidder auctions and dominant-strategy incentive compatibility (DSIC), generating certified revenue upper bounds. Our approach parametrizes Lagrange multipliers with a structurally guaranteed strict flow-conservation property using neural networks, enabling efficient optimization over feasible dual solutions via gradient descent. To bridge the gap between discrete computational methods and  theoretical guarantees for continuous types, we develop a novel lifting technique that maps dual certificates from coarse discretizations to fine refinements. We prove that lifting  gives valid revenue upper bounds for multi-item, multi-bidder auctions with continuous uniform valuations. Furthermore, we give a generalized lifting construction for arbitrary continuous distributions and demonstrate that these lifted duals  converge to the revenue of the original continuous problem in the discrete limit. We validate this  computational framework for the dual auction design problem by recovering known analytical mechanisms for canonical instances. For multi-item multi-bidder problems, our framework establishes a small gap between the optimal revenue and best-known DSIC mechanisms, providing  computational certificates of  near-optimality.
\end{abstract}

\begin{titlepage}

\maketitle\makeatletter \gdef\@ACM@checkaffil{} \makeatother
\setcounter{tocdepth}{2} %

\end{titlepage}

\section{Introduction}

Auctions remain among the most enduring and significant market institutions, and they underpin a wide range of modern economic activities, such as online advertising, the allocation of spectrum and energy procurement contracts, the issuance of treasury and other financial securities, and the pricing of computing resources in online markets. Their study, particularly the design of revenue-maximizing auctions, is a canonical problem in economic theory. 

The foundational contribution of~\citet{myerson1981optimal} characterizes the revenue optimal mechanism for selling a single item, but a comprehensive understanding of optimal auction design in general settings remains elusive. In the domain of {\em dominant-strategy incentive-compatible} (DSIC, or strategy-proof) mechanisms, existing closed-form characterizations are confined to variants of the single-bidder setting~\citep{daskalakis2015strong,manelli2006bundling,pavlov2011optimal, giannakopoulos2014duality}, where DSIC also coincides with Bayesian incentive compatibility (BIC), and to very limited settings for multiple bidders—most notably, settings with two items and distributions supported on two values~\citep{yao2017dominant}.

The use of computational techniques to design  mechanisms such as auctions, termed \emph{automated mechanism design} (AMD), was introduced by~\citet{conitzer2002,conitzer2004}. Early work on AMD advocated linear programming (LP) approaches in the case of agents with discrete type spaces~\citep{conitzer2002,conitzer2004,conitzer2006computational}. However, LP-based formulations scale exponentially in both the number of bidders and the number of items: with $k$ types per item and $n$ bidders facing $m$ items, the joint type space has size $k^{nm}$, and the LP requires a variable and a DSIC constraint for every profile.
More recently, \emph{differentiable economics}~\citep{dutting2024optimal,rahme2020auction,ivanov2022optimal,curry2022differentiable,duan2023scalable,shen2019automated,wang2025bundleflow}  proposes to make use of deep learning to discover revenue-maximizing auctions. In this framework, neural networks serve as flexible representations of allocation and payment rules, these rules optimized via sampling from a known type distribution and minimizing carefully designed objective functions.

Among the methods of differentiable economics, {\em GemNet}~\citep{wang2024gemnet} provides the first approach to auction design that simultaneously satisfies  three important properties: (1) \emph{expressiveness}, meaning the function class induced by the network architecture contains revenue-optimal auctions; (2) \emph{strategy-proofness}, ensuring the learned mechanisms satisfy DSIC; and (3) \emph{multi-bidder, multi-item applicability}. In this way, and recognizing they are DSIC, 
the auctions designed by GemNet provide a certified lower bound on the optimal revenue of multi-bidder, multi-item DSIC auctions.

Despite this progress, GemNet frames  revenue as an objective function to be optimized by gradient ascent (actually, descent in regard to negated revenue), and despite strong empirical performance (essentially matching theoretical benchmarks where they exist), it does not offer a formal guarantee of attaining the  optimal design. Achieving this kind of optimality would require the GemNet learning procedure to converge to a global maximizer of a highly non-convex non-smooth learning objective--a well-known limitation of modern deep learning methods~\citep{bottou2018optimization,Goodfellow-et-al-2016}.

This leaves a gap: how to compute a strong, valid upper bound on DSIC revenue in general, multi-bidder, multi-item DSIC auctions? The present paper develops the first computational framework for computing certificates on revenue for multi-item, multi-bidder DSIC auctions.
We compute a \emph{certified upper bound} on optimal revenue by studying the dual formulation from the perspective of deep learning augmented by a sound lifting scheme that is used to improve the bound that comes from computation on a finite grid.
Our computational approach leverages the duality framework for  discrete type, BIC optimal auctions~\citep{cai2016duality}, as generalized to the discrete type, multi-bidder DSIC setting,\footnote{Personal communication, S.~M.~Weinberg, E.~Xue, and E.~Ryu.}
and further extended here to the continuous type, multi-bidder DSIC setting.

We start by discretizing the continuous valuation distribution, concentrating probability mass onto a finite grid of bidder types. As established in prior work, 
a dual formulation for the optimal design problem 
is available for this multi-bidder, multi-item DSIC problem, 
this developed by Lagrangifying the DSIC constraints. The optimal Lagrange multipliers  give the virtual value of each item to each bidder. By strong duality, these virtual values give  the exact optimal revenue for the auction under the discretized distribution.
However, solving for these optimal Lagrange multipliers is challenging: (1) the number of  variables increases with the number of DSIC constraints, creating significant scalability issues; and (2) Lagrange multipliers have to satisfy a flow conservation property to ensure the objective function of the dual is finite (and violating these flow constraints even slightly invalidates the dual certificate).

To address these challenges, we learn the Lagrange multipliers via a novel deep learning architecture that enforces flow conservation as a structural property of the architectural design. Rather than representing multipliers directly, which would require enforcing constraints on a large output space, the neural network design learns a state-dependent routing policy, where the state is the valuation profile $v_{-i}$ of the other bidders. By interpreting this policy as defining an absorbing Markov chain on the grid of types, we can solve  by construction for the unique flows that satisfy conservation. This ensures that the neural network produces a valid, flow-conserving dual solution, and therefore a valid  upper bound on optimal revenue. In this way, we convert the search for a validity-constrained certificate into an unconstrained minimization problem that can be optimized directly via standard gradient-based methods.

The remaining challenge is to bridge from a computational framework on a discretized type distribution to a valid revenue bound for optimal auction design on a continuous type distribution. For uniform value distributions, we introduce a \emph{lifting technique} that rigorously maps the dual certificate on a discretized version of the problem to the auction design problem on the continuous domain, yielding a sequence of decreasing revenue upper bounds as the grid is refined. For general value distributions (e.g., the beta distribution), we do not have an analogous lifting construction; instead, we prove that the discretized dual objective converges to a valid upper bound on the continuous optimal revenue in the limit of grid resolution.

We demonstrate the effectiveness of our method across a range of settings with independent, additive valuations, providing the first rigorous, near-optimal revenue certificates in multi-bidder settings.
For $2\times 2$ and $3\times 2$ ($n$ bidders $\times$ $m$ items)
settings, we achieve certified upper bounds that are within 1.8\% and 3.7\%, respectively, of the best known revenue for the primal problem  as given by GemNet~\cite{wang2024gemnet}. 
This gives the first formal link between these primal solutions and the theoretical optimum
and proves that state-of-the-art  mechanisms obtained through deep learning
are  near-optimal. Conversely, we show that primal baselines with approximate incentive compatibility, such as RegretNet~\citep{dutting2024optimal}, can yield revenues higher than our certified upper bound, confirming for the first time that
they  overestimate the 
optimal revenue  as a result of their incentive violations.

We also validate our dual framework in simpler, discrete type settings by exactly recovering the optimal revenue for the only known, analytical, multi-bidder, multi-item optimal auction~\citep{yao2017dominant}. The dual computational framework also derives a bound within 0.6\% of the analytically optimal revenue for the continuous type, single-bidder bundling setting of \citet{manelli2006bundling}.

\subsection{Related Work}

The application of deep learning to mechanism design, often termed \emph{differentiable economics}, has focused primarily on the primal problem: learning the allocation and payment rules directly. \citet{dutting2024optimal} introduced RegretNet, which models the mechanism as a neural network and minimizes expected regret to approximate incentive compatibility (IC). While successful empirically, RegretNet only satisfies IC approximately and does not provide revenue guarantees. The lack of regret guarantees motivated works \citep{curry2020certifying} that aims to certify the regret using integer program methods. However, the primal nature of these directions precludes a sufficient understanding of optimality, and because RegretNet allows incentive violations, it remains unclear whether the reported revenue by RegretNet is higher or lower than the theoretical optimum.

\citet{wang2024gemnet} proposed GemNet, which enforces exact strategy-proofness (DSIC) by using the multi-bidder menu formulation and a mixed integer linear programming post-processing step. While GemNet provides a  certified lower bound on the optimal revenue (by exhibiting a valid primal mechanism), it cannot certify how close revenue is to the true optimum. 

Early work in automated mechanism design relied on linear programming (LP) to design optimal mechanisms for discrete type spaces~\citep{conitzer2002, conitzer2004, guo2010computationally, sandholm2015automated}. These approaches solve the primal problem directly. While effective for small instances, they are fundamentally limited by the combinatorial explosion of the type space in multi-item, multi-bidder settings. The tabular nature of LPs require optimizing a distinct variable for every possible valuation profile. As the type space grows, the number of required variables and constraints exceeds the capacity of modern hardware and is intractable.

Our approach overcomes this challenge by parameterizing the  variables in the dual representation of the optimal auction problem via a neural network. Instead of storing a global lookup table of  variables, we represent the dual solution as a continuous mapping from valuation profiles to flow variables. This allows for the number of optimization parameters to be completely decoupled from the size of the type space. Furthermore, rather than striving to satisfy flow constraints that make a dual solution suitable for an upper bound, our architecture enforces flow conservation by construction. This transforms what would be an intractable constrained problem into a lower-dimensional unconstrained optimization that can be solved efficiently using stochastic gradient descent.

Our theoretical framework builds upon the existing duality theory of  optimal auction design. For the single-bidder, multi-item setting, \citet{daskalakis2015strong} and \citet{giannakopoulos2014duality} establish strong duality frameworks  through the lens of optimal transport. For multi-bidder settings and Bayesian incentive compatibility (BIC), \citet{cai2016duality} provide a unified duality framework for discrete type spaces, characterizing the optimal revenue via flow conservation constraints on the type graph.

Recent work has explored extending flow-based formulations to continuous domains. \citet{barber2020bounding} introduces continuous dual flow constructions for multi-bidder, additive buyers and two items, with continuous valuations and BIC,
while \citet{XYZ} extends this approach, also exploiting symmetries, to settings with multiple items (still BIC). In a related direction, \citet{kolesnikov2022beckmann} leverage continuous optimal transport to formulate a dual for multi-item, multi-bidder, BIC auctions. 
In addition to being BIC and not DSIC, 
the numerical solutions in~\citet{kolesnikov2022beckmann}, which  
apply to a discretized  formulation, 
do not establish either a rigorous upper or a rigorous lower bound on
the optimal continuous revenue. In contrast, we focus on the multi-bidder DSIC
regime and  take as our starting point a ``partial dual" formulation (that Lagrangifies only incentive constraints to preserve the structure of virtual welfare maximization) 
for a discrete type space,\footnote{Personal communication, S.~M.~Weinberg, E.~Xue, and E.~Ryu.}  adapting the flow conservation constraints of \citet{cai2016duality} from the BIC to the DSIC setting.

We also distinguish our  computational framework from prior theoretical results. \citet{yao2017dominant} analyzed DSIC vs. BIC gaps by deriving exact analytical formulas for distributions with binary support. While these results offer deep theoretical insights, they rely on the binary support condition.
\citet{zuo2017generalizing} explores a ``complete dual" formulation for multi-dimensional auctions. Unlike the ``partial dual" approach,
the complete dual also Lagrangifies the allocation feasibility constraints, resulting
in a formulation that is less amenable to deriving  certificates of near-optimality. Our work combines the partial dual setup with deep learning to compute certificates for  high-dimensional instances where analytical construction is infeasible.

\section{Preliminaries}\label{sec:pre}

\textbf{Sealed-Bid Auctions, DSIC, and IR}.\ \ 
We study sealed-bid auctions with $n$ bidders and $m$ items, and each bidder $i$ has a \emph{valuation function} $v_i:2^M\to\mathbb{R}\ge 0$. Denote the sets of bidders and items by $N$ and $M$, respectively. We focus on \emph{additive valuations}. Bidder $i$’s value for any subset $S\subseteq M$ is $v_i(S)=\sum_{j\in S} v_i({j})$, where $v_i({j})$ denotes the value of item $j\in M$ to bidder $i$. The valuation $v_i$ is drawn independently from a continuous distribution $D_i$ with CDF $F_i$ and PDF $f_i$. Correspondingly, the joint distribution is $D$, with CDF $F$, and PDF $f(v)=\prod_{i=1}^n f_i(v_i)$. $D_i$ may be an arbitrary joint distribution on $[0,v_{\max}]^m$; we require independence across bidders, but allow arbitrary correlation across items within a bidder. We assume valuations are bounded: for all $i$ and $j$, $v_i(j)\in[0,v_{\max}]$ with $v_{\max}>0$. Define value profile $v=(v_1,\dots,v_n)\in [0,v_{\max}]^{m \times n} =V$. Here, $V$ is the space of feasible valuation functions. We denote the value space of bidder $i$ by $V_i$. Subscripts $-i$, e.g., $v_{-i}$, denote all agents other than $i$.

A mechanism is specified by an allocation rule $x(\cdot)$ and a payment rule $p(\cdot)$. For each bidder $i$ and item $j$, $x_{ij}(v)$ is the probability that $i$ receives item $j$ under profile $v$, and $p_i(v)$ is the payment charged to bidder $i$. The mechanism is \emph{Dominant Strategy Incentive Compatible (DSIC)} if, for every bidder $i$, all $v_i$ and $v_i'$, and all $v_{-i}$,
\[
\sum_{j=1}^m x_{ij}(v)\cdot v_{ij}-p_i(v)\ge
\sum_{j=1}^m x_{ij}(v_i',v_{-i})\cdot v_{ij}-p_i(v_i',v_{-i}).
\]
In other words, reporting truthfully maximizes bidder $i$'s utility irrespective of others' reports. The mechanism is \emph{Individually Rational (IR)} if for all bidders $i$ and all profiles $v$,
\(
\sum_{j=1}^m x_{ij}(v)\cdot v_{ij}-p_i(v)\ge 0,
\)
so participation weakly benefits every bidder. For a DSIC/IR mechanism $(x,p)$, let $\mathrm{Rev}_{(x,p)}(D)$ denote its expected revenue under $D$, and let $\mathrm{Rev}(D)$ be the optimal revenue achievable by any DSIC/IR mechanism.

\begin{wrapfigure}[11]{r}{0.45\textwidth}
  \centering
  \vspace{-\baselineskip} %
  \includegraphics[width=\linewidth]{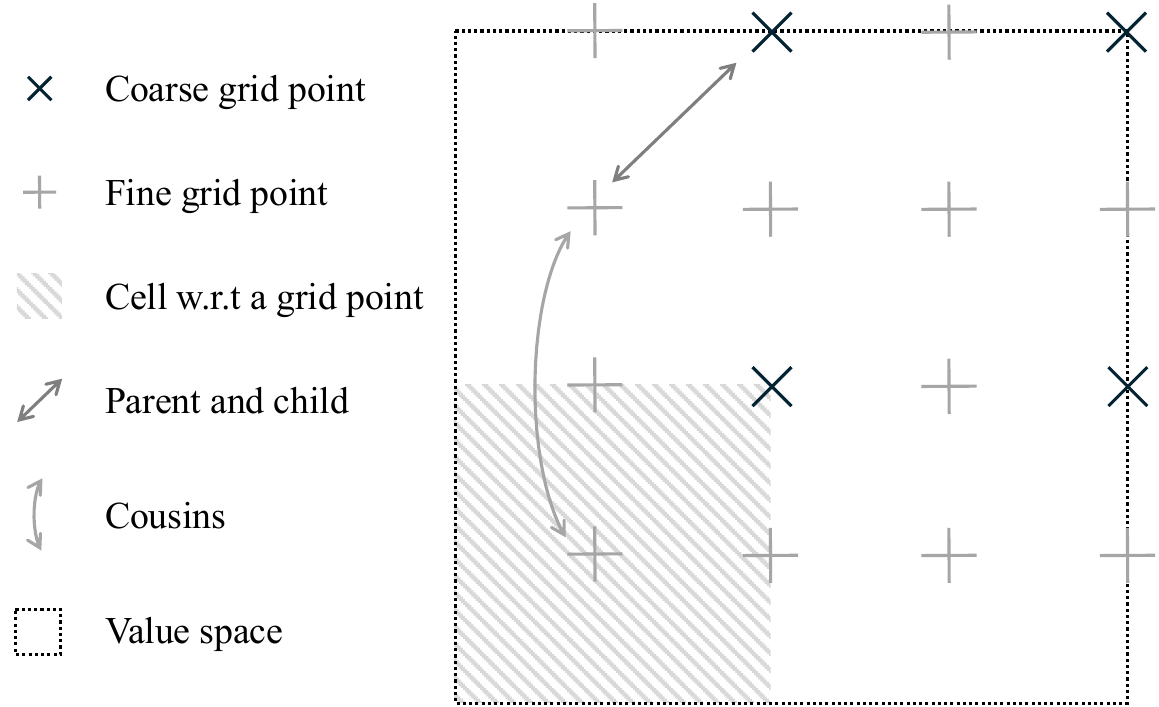}
  \caption{Our grid, refinement, and cell scheme.}
  \label{fig:legend}
\end{wrapfigure}
\noindent \textbf{Discretization and Refinement}.\ \ 
Let $d = n\times m$.
For an integer $\rho\ge 1$, set the mesh size $\Delta := v_{\max}/\rho$ and define the (coarse) uniform product grid
\(
\mathcal{G}_{\Delta}
\;:=\;
\{\Delta,2\Delta,\ldots,v_{\max}\}^{d}.
\)

\emph{Cell.} Define the coordinatewise rounding-up map $R_{\Delta}:[0,v_{\max}]^{d}\to\mathcal{G}_{\Delta}$ by
$(R_{\Delta}(v))_{ij} =\Delta\Big\lceil \frac{v_{ij}}{\Delta}\Big\rceil$, $i\in [n]$, $j\in[m]$.
Then $v \le R_{\Delta}(v)$, and $R_{\Delta}(v)$ is the unique grid point in
$\mathcal{G}_{\Delta}$ that is coordinatewise minimal among those dominating $v$. Define the cell of a grid point $w$ as $C(w)=\{v\in V|R_{\Delta}(v)=w\}$. 

As a discretization of the continuous distribution $f$, we denote the distribution over $\mathcal{G}_{\Delta}$ as $f^{\mathcal{G}_{\Delta}}$.
We have
\begin{align}
    f^{\mathcal{G}_{\Delta}}(w) = \int_{v\in C(w)} f(v) \mathrm{d}v, \qquad \forall w\in\mathcal{G}_{\Delta}.
\end{align}

\emph{Refinement}.
The {\em $K$-refined grid} (for integer $K\ge 2$) is obtained by subdividing each
coordinate interval into $K$ equal parts:
\(
\mathcal{G}_{\Delta/K}
\;:=\;
\{\tfrac{\Delta}{K},\tfrac{2\Delta}{K},\ldots,v_{\max}\}^{d}.
\)
Equivalently, the refinement factor $K$
decreases the mesh size from $\Delta$ to $\Delta/K$ in every coordinate and $\mathcal{G}_{\Delta}\subset \mathcal{G}_{\Delta/K}$.
We analogously define $R_{\Delta/K}$ by replacing $\Delta$ with $\Delta/K$.

\noindent \textbf{Dual, Lagrange, and Flow-Conservation.}
We first consider the problem of DSIC auction design on a discrete grid. For clarity, we denote the grid by $\mathcal{G}$ and the discretized distribution on it by $D^\mathcal{G}$, with PDF being $f^\mathcal{G}$. For this discrete setting, we derive an upper bound on $\mathrm{Rev}(D^\mathcal{G})$---the optimal revenue achievable by any DSIC/IR mechanism under the discretized distribution $D^\mathcal{G}$---based on a known duality framework.\footnote{Personal communication, S.~M.~Weinberg, E.~Xue, and E.~Ryu.} We will later (\cref{sec:method:lift,sec:method:continuous}) relate this discrete quantity to the continuous optimal revenue $\mathrm{Rev}(D)$.
For every bidder $i$ and every $v\in \mathrm{supp}(D^\mathcal{G})$, we introduce decision variables $x_{ij}(v)$ and $p_i(v)$. We include a special type $\varnothing$ representing non-participation, and enforce that for all bidders $i$ and all $v_{-i}$, $x_{ij}(\varnothing,v_{-i})=0$ for each item $j$ and $p_i(\varnothing,v_{-i})=0$. With this modeling choice, IR constraints can be expressed as DSIC constraints by allowing a bidder to misreport $\varnothing$. The resulting linear program is:
\begin{align}
& \text{Variables: } x_{ij}(v), p_i(v)\ \forall i\in [n], j\in [m], v\in \mathrm{supp}(D^\mathcal{G}) \nonumber\\
& \text{Maximize: } \sum_{v\in\mathcal{G}} f^\mathcal{G}(v)\cdot \sum_i p_i(v)\\
& \text{Subject to: } \sum_j x_{ij}(v)\cdot v_{ij}-p_i(v)\ge
\sum_j x_{ij}(v_i',v_{-i})\cdot v_{ij}-p_i(v_i',v_{-i}) \tag{DSIC constraints} \nonumber\\
&\qquad \qquad \qquad \qquad \forall i\in [n], v_i\in \mathrm{supp}(D^\mathcal{G}_i), v_i'\in \mathrm{supp}(D^\mathcal{G}_i)\cup\{\varnothing\},v_{-i}\in \mathrm{supp}(D^\mathcal{G}_{-i})\nonumber\\
& \qquad \qquad \ \ \sum_i x_{ij}(v)\le 1 \quad \forall j\in [m],\ v\in \mathrm{supp}(D^\mathcal{G}), && \tag{Feasibility constraints} \nonumber\\
& \qquad \qquad \ \ \ x_{ij}(\varnothing,v_{-i})=0;p_i(\varnothing,v_{-i})=0  \quad \forall i\in [n], j\in [m], v_{-i}\in \mathrm{supp}(D^\mathcal{G}_{-i}),\nonumber \\
& \qquad \qquad \ \ \ x_{ij}(v)\ge 0 \quad \forall i\in [n], j\in [m], v\in \mathrm{supp}(D^\mathcal{G}).\nonumber
\end{align}

Let $P_+(D^\mathcal{G})$ denote the {\em polytope of feasible solutions}. We dualize only the DSIC constraints. Let $P(D^\mathcal{G})$ be the polytope defined by the remaining constraints, and let $\lambda^{v_{-i}}_i(v_i,v_i')$ be the Lagrange multiplier for the DSIC constraint asserting that bidder $i$ with true type $v_i$ does not prefer to report $v_i'$ when others report $v_{-i}$. The partial Lagrangian, $\mathcal{L}^{\mathcal{G}}(\lambda,x,p)$, is
\begin{align*}
\mathcal{L}^\mathcal{G}(\lambda,x,p)
&=
\sum\nolimits_{\substack{i\in [n]\\ v\in \mathrm{supp}(D^\mathcal{G})}} f^\mathcal{G}(v)\cdot p_i(v) \\
&\quad +
\sum\nolimits_{\substack{i\in [n], v\in \mathrm{supp}(D^\mathcal{G})\\ v_i'\in \mathrm{supp}(D^\mathcal{G}_i)\cup\{\varnothing\}}}
\lambda^{v_{-i}}_i(v_i,v_i')
\Bigl(\sum_j v_{ij}\cdot (x_{ij}(v)-x_{ij}(v_i',v_{-i}))-(p_i(v)-p_i(v_i',v_{-i}))\Bigr) \\
&\hspace{-4em} =
\sum\nolimits_{i\in [n],  v\in \mathrm{supp}(D^\mathcal{G})} p_i(v)
\Bigl(
f^\mathcal{G}(v)+\sum\nolimits_{v_i'\in \mathrm{supp}(D^\mathcal{G}_i)} \lambda^{v_{-i}}_i(v_i',v_i)
-\sum\nolimits_{v_i'\in \mathrm{supp}(D^\mathcal{G}_i)\cup\{\varnothing\}} \lambda^{v_{-i}}_i(v_i,v_i')
\Bigr) \\
&\hspace{-4em}\quad +
\sum\nolimits_{\substack{i\in [n], j\in [m]\\ v\in \mathrm{supp}(D^\mathcal{G})}} x_{ij}(v)
\Bigl(
\sum\nolimits_{v_i'\in \mathrm{supp}(D^\mathcal{G}_i)\cup\{\varnothing\}} \lambda^{v_{-i}}_i(v_i,v_i')\cdot v_{ij}
-\sum\nolimits_{v_i'\in \mathrm{supp}(D^\mathcal{G}_i)} \lambda^{v_{-i}}_i(v_i',v_i)\cdot v_{ij}'
\Bigr).
\end{align*}

\emph{Flow conservation.} For a fixed bidder $i$ and a fixed $v_{-i}$, we say $\lambda^{v_{-i}}_i$ satisfies \emph{flow conservation} if
\begin{align}
    \sum_{v_i'\in \mathrm{supp}(D^\mathcal{G}_i)\cup\{\varnothing\}} \lambda^{v_{-i}}_i(v_i,v_i')
=
f^\mathcal{G}(v)+\sum_{v_i'\in \mathrm{supp}(D^\mathcal{G}_i)} \lambda^{v_{-i}}_i(v_i',v_i), \label{equ:fc}
\end{align}
for every $v_i\in \mathrm{supp}(D^\mathcal{G}_i)$. Equivalently, this describes conservation in a flow network with a super source $s$, a super sink $\varnothing$, nodes indexed by $v_i\in \mathrm{supp}(D^\mathcal{G}_i)$, edges $s\to v_i$ carrying $f^\mathcal{G}(v_i,v_{-i})$, and edges $v_i\to v_i'$ carrying $\lambda^{v_{-i}}_i(v_i,v_i')$ for $v_i'\in \mathrm{supp}(D^\mathcal{G}_i)\cup\{\varnothing\}$. A dual solution $\lambda$ is flow-conserving if this holds for every bidder $i$ and every $v_{-i}$. Flow conservation is important due to the following lemma. We give the proof  in~\cref{appx:proof}.
\begin{restatable}[Flow Conservation]{lemma}{flowconservation}
    $\max_{(x,p)\in P(D^\mathcal{G})} \mathcal L^\mathcal{G}(\lambda,x,p) < \infty$ iff $\lambda$ satisfies flow conservation.
\end{restatable}

\noindent \textbf{Virtual values}.
Given any flow-conserving $\lambda$, define a virtual value function $\Phi^\lambda$ by
\begin{equation} \label{eq:virtual_value_def}
\Phi^\lambda_{ij}(v) := v_{ij}-
\frac{1}{f^\mathcal{G}(v)}\sum_{v_i'\in \mathrm{supp}(D^\mathcal{G}_i)} \lambda^{v_{-i}}_i(v_i',v_i)(v_{ij}'-v_{ij}).
\end{equation}

\begin{restatable}[Virtual Welfare Upper Bounds Revenue]{lemma}{vw} If $\lambda$ satisfies flow conservation, then for every $(x,p)\in P_+(D^\mathcal{G})$, the revenue is upper bounded by the virtual welfare of $x$ w.r.t. $\Phi^\lambda$, i.e.,
\[
\sum_{v\in \mathrm{supp}(D^\mathcal{G})} f^\mathcal{G}(v)\cdot \sum_i p_i(v)
\le
\sum_{v\in \mathrm{supp}(D^\mathcal{G})} f^\mathcal{G}(v)\cdot \sum_i x_i(v)\cdot \Phi^\lambda_i(v).
\]
Moreover, equality holds precisely when, for all $i$, $v_i$, $v_i'$, and $v_{-i}$ with $\lambda^{v_{-i}}_i(v_i,v_i')>0$, the associated DSIC constraint is tight. If $(x^*,p^*)$ is an optimal DSIC/IR mechanism and $\lambda^*$ is an optimal dual solution, then the revenue of $(x^*,p^*)$ equals the virtual welfare of $x^*$ under $\Phi^{\lambda^*}$, and
\[
x^* \in \arg\max_{(x,p)\in P(D^\mathcal{G})}
\sum_{v\in \mathrm{supp}(D^\mathcal{G})} f^\mathcal{G}(v)\cdot \sum_i x_i(v)\cdot \Phi^{\lambda^*}_i(v).
\]
\end{restatable}

This yields the dual representation
\begin{equation}
\label{eq:dual}
\Rev( D^\mathcal{G})
~=
\min_{\lambda~\text{flow-conserving}}
\; \E_{v\sim  D^\mathcal{G}}\Biggl[\max_{x\in P( D^\mathcal{G})}
\sum_{i=1}^n\sum_{j=1}^m x_{ij}(v) \Phi_{ij}^{\lambda}(v)\Biggr].
\end{equation}
Because $x$ decouples across items and $\sum_i x_{ij}(v)\le 1$ with $x_{ij}(v)\ge 0$,
for every fixed $(\lambda,v)$, the maximizer allocates item $j$ to a bidder with the
largest (nonnegative) virtual value:
\begin{equation}
\label{eq:maxx}
\max_{x\in P( D^\mathcal{G})}\sum_{i=1}^n\sum_{j=1}^m x_{ij}(v) \Phi_{ij}^{\lambda}(v)
=\sum_{j=1}^m \max\{0,~\max_{i\in[n]}\Phi_{ij}^{\lambda}(v)\}=\sum_{j=1}^m (\max_{i\in[n]}\Phi_{ij}^{\lambda}(v))_+.
\end{equation}

\section{Methodology}
\label{sec:method}

The computational framework that we develop first derives a certified upper bound on $\mathrm{Rev}(D^\mathcal{G})$ for a discrete-type problem by minimizing the dual objective derived in \cref{eq:dual} (\cref{sec:method:flow}). For uniform value distributions, we then introduce a  lifting technique (\cref{sec:method:lift}) that, given a feasible dual solution on a coarse grid, $\mathcal{G}^\mathrm{C}$, produces a feasible dual solution on any refinement, $\mathcal{G}^\mathrm{F}$, whose dual objective is no larger. This makes the resulting sequence of bounds, $\mathrm{Rev}(D^\mathrm{F})$,  monotone non-increasing as the grid is refined. In \cref{sec:method:continuous}, we connect these discrete bounds to the continuous optimal revenue, $\mathrm{Rev}(D)$. For uniform distributions, the lifting construction extends directly to the continuous domain, yielding a rigorous upper bound on $\mathrm{Rev}(D)$ from any coarse-grid solution (\cref{Thm:ct}). For general value distributions, we do not have an analogous lifting construction. Instead, we prove that the discretized dual objective, plus an explicit mesh-dependent correction term, yields a valid upper bound on $\mathrm{Rev}(D)$, with the correction vanishing as $\Delta \to 0$ (\cref{thm:coarse-plus-error}).

\subsection{Neural Network Flow Parameterization via Absorbing Markov Chains}
\label{sec:method:flow}

We first discuss how to compute a certified upper bound for a discrete-type problem on a fixed grid, $\mathcal{G}$, with PMF $f^\mathcal{G}$. The question of how such a grid relates to an underlying continuous distribution is deferred to \cref{sec:method:lift}.

The first challenge is the scalability issue, as the number of Lagrange multipliers increases with the discretization resolution.
To deal with this, we generate the Lagrange multipliers using a neural network $g_\theta: [0,v_{\max}]^{(n-1)m}\rightarrow \mathbb{R}^{(T'+1) \times T'}$, where $\theta$ is learnable parameters and $T'$ is the number of types. The neural network's goal is to learn the optimal dual variables $\lambda_i^{v_{-i}}(v_i,v_i')$. $g_\theta$ takes as input the value profile of other agents $v_{-i}$ and outputs the raw parameters defining a flow routing policy. Specifically, the network predicts a matrix of transition logits $L \in \mathbb{R}^{T' \times T'}$, determining the transition probabilities
between two discretized types $(u, w)$, and a vector of sink logits $s \in \mathbb{R}^{T'}$, determining the transition probabilities
from each type $u$ to $\varnothing$.

We do not directly output Lagrange multipliers because of the strict flow conservation requirement (\cref{equ:fc}). The dual variables $\lambda$ must form a valid, flow-conserving solution for every bidder $i$ and every $v_{-i}$. To strictly guarantee flow conservation, we enforce this property structurally by interpreting the neural network outputs, $L$ and $s$, as absorbing Markov chains.

Specifically, for each type node $u$ in the grid $\mathcal{G}_i$, we apply a softmax over the transition logits $L_{u,\cdot}$ (corresponding to the $T'$ outgoing type-to-type edges) and the sink logit $s_u$ (corresponding to the additional ``sink" edge).
This yields transition probabilities $\Pi_i^{v_{-i}}(u,w)$ for all $w\in\mathcal{G}_i$ and sink-transition probabilities $\alpha_i^{v_{-i}}(u)$, satisfying
\[
\sum_{w\in \mathcal{G}_i}\Pi_i^{v_{-i}}(u,w) \;=\; 1-\alpha_i^{v_{-i}}(u), \qquad \forall u \in \mathcal{G}_i.
\]
To ensure that $\lambda$ can be solved, we enforce a uniform leakage condition:
$\alpha_i^{v_{-i}}(u)\ge \varepsilon$ for all $u$, for a small constant $\varepsilon>0$. Under this condition, every row of $\Pi_i^{v_{-i}}$ sums to at most $1-\varepsilon$, so the chain on $\mathcal G_i$ is absorbing (all mass exits to the sink with probability $1$). Equivalently, $\Pi_i^{v_{-i}}$ is a \emph{sub-stochastic} transition matrix on $\mathcal G_i$.

This parameterization allows us to solve for a unique, feasible flow consistent with the routing policy generated by $g_\theta$.
Fix bidder $i$ and $v_{-i}$.
Let $\nu^{v_{-i}}\in\mathbb R^{T'}$ denote the vector of super-source inflows, whose entries are $f^{\mathcal{G}}(u,v_{-i})$, for each $u\in \mathcal G_i$.
Let $\mu^{v_{-i}}\in\mathbb R^{T'}$ be the vector of node outflows, whose entries are the total flow routed out of each $u\in \mathcal G_i$.
Flow conservation requires that total outflow equals super-source inflow plus routed inflow from other nodes:
\[
\mu^{v_{-i}} \;=\; \nu^{v_{-i}} + \left(\Pi_i^{v_{-i}}\right)^\top \mu^{v_{-i}}.
\]
Since the chain is absorbing (every row of $\Pi_i^{v_{-i}}$ sums to at most $1-\varepsilon$, $0<\varepsilon<1$), the matrix $I-(\Pi_i^{v_{-i}})^\top$ is invertible. We thus obtain the unique, non-negative outflow vector $\mu$:
\begin{equation}
\mu^{v_{-i}}
=
\left(I - (\Pi_i^{v_{-i}})^\top\right)^{-1} \nu^{v_{-i}}.
\label{eq:linear_solve}
\end{equation}

Finally, the explicit edge flows are recovered by scaling routing probabilities by node outflow:
for $w\in\mathcal G_i$, we set $\lambda_i^{v_{-i}}(u,w)=\mu^{v_{-i}}(u)\,\Pi_i^{v_{-i}}(u,w)$, and the super-sink flow is
$\lambda_i^{v_{-i}}(u,\varnothing)=\mu^{v_{-i}}(u)\,\alpha_i^{v_{-i}}(u)$. This construction guarantees that the resulting dual variables $\lambda$ satisfy flow conservation exactly (up to floating-point precision) by construction. By backpropagating through the linear system solution in \cref{eq:linear_solve}, we can train the network parameters $\theta$ to minimize the dual objective while remaining on the manifold of feasible dual solutions.

\subsection{Stochastic Training Objective}\label{sec:method:training}

During the training phase, our objective is to minimize the dual upper bound $J(\lambda_\theta)$. For multi-bidder auctions, the size of the joint type space grows exponentially with the number of bidders, making exact summation intractable. We therefore employ stochastic optimization. Unlike tabular LP methods that scale with the size of the profile space, our neural network parameterization scales with the complexity of the routing function, allowing us to handle type spaces that are inaccessible to classical solvers. In each training step, we sample a mini-batch of valuation profiles $v \sim D^\mathcal{G}$. For each sampled profile, we compute the feasible flow $\lambda$ using the linear system described in \cref{sec:method:flow}. We then calculate the resulting virtual values $\Phi^\lambda(v;\theta)$ according to \cref{eq:virtual_value_def}.

The loss function is defined as the empirical estimate of the expected virtual welfare. Specifically, for a batch $\mathcal{B}$, we minimize
\begin{equation}
    \mathcal{L}(\theta) = \frac{1}{|\mathcal{B}|} \sum_{v \in \mathcal{B}} \sum_{j=1}^m \max \left( 0, \max_{k \in [n]} \Phi^\lambda_{kj}(v;\theta) \right).
\end{equation}

While this phase relies on stochastic approximation to update the weights, the ``hard constraint'' architecture ensures that even solutions during the early stages of training are valid dual certificates. A suboptimal network may yield a loose upper bound, but it will never produce an invalid one.

\subsection{Tightened Certification via Lifting for Uniform Distribution}
\label{sec:method:lift}

We now consider a continuous value distribution, $D$ (assumed uniform throughout this subsection) and two  different discretizations: a coarse grid, $\mathcal{G}^\mathrm{C}$, and a $K$-fold refinement, $\mathcal{G}^\mathrm{F}$, constructed as in \cref{sec:pre}. Let $D^\mathrm{C}$ and $D^\mathrm{F}$ denote the discretizations induced on these grids, with PMFs $f^\mathrm{C}$ and $f^\mathrm{F}$. Since refining the grid introduces computational costs---a refinement factor $K$ multiplies the joint type space by $K^{nm}$—we cannot directly solve the dual at arbitrarily fine resolution. Lifting offers a workaround: given a feasible dual solution $\lambda^\mathrm{C}$ on $\mathcal{G}^\mathrm{C}$, we construct a feasible $\lambda^\mathrm{F}$ on $\mathcal{G}^\mathrm{F}$ whose dual objective is never worse. Under the uniform distribution, this yields a sequence of optimal revenue, $\mathrm{Rev}(D^\mathrm{F})$, that decreases monotonically as the grid is refined.

Because $D$ is uniform, the induced PMF is constant on the coarse grid; i.e., $f^{\mathrm{C}}(w_a)=f^{\mathrm{C}}(w_b)$ for every pair of coarse types $w_a, w_b$. We now define the parent map.

\begin{definition}[Parent Map]
    Every $v\in \mathcal{G}^\mathrm{F}$ has a unique \emph{parent}, $R(v)\in \mathcal{G}^{\mathrm{C}}$, obtained by rounding
to the coordinatewise minimal
coarse grid point dominating $v$.  We write $w=R(v)$, and have $w\ge v$ coordinate-wise. We write $S(w)=\{v|v\in \mathcal{G}^\mathrm{F}, R(v)=w\}$ as the set of child grid points of $w$. This parent-child relationship and the grid structure are illustrated in \cref{fig:legend}.
\end{definition}

For every coarse type $w$ and a fine type $v^{(k)} = w - \delta^{(k)} \in S(w)\subset \mathcal{G}^\mathrm{F}$, it follows that  $f^{\mathrm{F}}(v^{(k)}) = \frac{1}{K^{mn}} f^{\mathrm{C}}(w)$, where $k$ is an index and $\delta^{(k)} \ge 0$ is a constant for all $w$.  Moreover, the following \emph{pushforward identity} holds:
\begin{equation}
\label{eq:push_n}
f^{\mathrm{C}}(w)
=\sum_{v\in S(w)} f^\mathrm{F}(v),
\qquad\forall\,w\in \mathcal{G}^{\mathrm{C}},
\end{equation}
for $v=(v_i,v_{-i})$, $w=(w_i,w_{-i})$, and writing $w_i=R(v_i)$ and
$w_{-i}=R(v_{-i})$ with a slight abuse of notation.
\begin{definition}[Cousin]
    $v\in \mathcal{G}^{\mathrm{F}}$ is a \emph{cousin} of $v'\in \mathcal{G}^{\mathrm{F}}$ if and only if $v$ and $v'$ are of the same relative position with respect to their parents; i.e., $v$ and $v'$ have the same index $k$ (see \cref{fig:legend}).
\end{definition}

\subsubsection{Lifting}

Suppose that we have a solution $\lambda_{i}^{\mathrm{C}}$ on a coarse grid. For every coarse edge $\lambda_{i}^{\mathrm{C},w_{-i}}(w_b, w_a)$, we create $K^m$ fine edges:
\begin{align}
    \lambda_{i}^{\mathrm{F},v_{-i}}(v_b^{(k)}, v_a^{(k)}) := \frac{1}{K^{mn}} \lambda_{i}^{\mathrm{C},w_{-i}}(w_b, w_a), \quad  R(v_a^{(k)})=w_a, R(v_b^{(k)})=w_b, k \in \{1, \dots, K^m\}.    \label{eq:lift-type_uniform}
\end{align}

The sink flow in the fine grid is constructed similarly: 
\begin{align}
    \lambda_{i}^{\mathrm{F},v_{-i}}(v_a^{(k)}, \varnothing) := \frac{1}{K^{mn}} \lambda_{i}^{\mathrm{C},w_{-i}}(w_a, \varnothing).    \label{eq:lift-sink_uniform}
\end{align}

\begin{lemma}[Fine flow conservation]\label{lem:flow-conservation_uniform}
The $\lambda^{\mathrm{F}}$ defined in Eq.~\ref{eq:lift-type_uniform} and Eq.~\ref{eq:lift-sink_uniform}
satisfies flow conservation for $D^{\mathrm{F}}$, i.e., for every $i$, fine
$v_{-i}$ and $a_i\in\supp(D^{\mathrm{F}}_i)$,
\[
\sum_{b_i\in\supp(D^{\mathrm{F}}_i)\cup\{\varnothing\}}
\lambda_{i}^{\mathrm{F},\,v_{-i}}(a_i,b_i)
\;=\;
f^{\mathrm{F}}(a_i,v_{-i})
\;+\;
\sum_{b_i\in\supp(D^{\mathrm{F}}_i)}
\lambda_{i}^{\mathrm{F},\,v_{-i}}(b_i,a_i).
\]
\end{lemma}

\begin{proof}
Fix any child $v_a^{(k)}, k \in \{1, \dots, K^m\}$. $v_a^{(k)}$ only receives flow from the $k$-th children   of each of source coarse grid points $w_b$:
\begin{align}
    \text{In}^{\mathrm{F}}(v_a^{(k)}) = \sum_{w_b} \lambda_{i}^{\mathrm{F},v_{-i}}(v_b^{(k)}, v_a^{(k)}) = \frac{1}{K^{mn}} \sum_{w_b} \lambda_{i}^{\mathrm{C},w_{-i}}(w_b, w_a) = \frac{1}{K^{mn}} \text{In}^{\mathrm{C}}(w_a).
\end{align}

Meanwhile, $v_a^{(k)}$ only sends flow to the $k$-th children of each of target grid point $w_c$ (or the sink):
\begin{align}
    \text{Out}^{\mathrm{F}}(v_a^{(k)}) = \sum_{w_c} \lambda_{i}^{\mathrm{F},v_{-i}}(v_a^{(k)}, v_c^{(k)}) + \lambda_{i}^{\mathrm{F},v_{-i}}(v_a^{(k)}, \varnothing) = \frac{1}{K^{mn}} \text{Out}^{\mathrm{C}}(w_a).
\end{align}
By coarse  flow conservation, this equals $\frac{1}{K^{mn}} f^{\mathrm{C}}(w_a)$, which is exactly $f^{\mathrm{F}}(v_a^{(k)})$:
\begin{align}
    \text{Out}^{\mathrm{F}}(v_a^{(k)}) - \text{In}^{\mathrm{F}}(v_a^{(k)}) &= \frac{1}{K^{mn}} \left( \text{Out}^{\mathrm{C}}(w_a) - \text{In}^{\mathrm{C}}(w_a) \right) = \frac{1}{K^{mn}} f^{\mathrm{C}}(w_a) = f^{\mathrm{F}}(v_a^{(k)}).
\end{align}
\end{proof}

\begin{lemma}[Virtual Value Pointwise Dominance]\label{lem:dominance_uniform}
    With $\lambda^{\mathrm{F}}$ defined in Eq.~\ref{eq:lift-type_uniform} and Eq.~\ref{eq:lift-sink_uniform},
    \begin{align}
        \Phi_{ij}^{\lambda^{\mathrm{F}}}(v_a^{(k)}) \le \Phi_{ij}^{\lambda^{\mathrm{C}}}(w_a), \quad \forall\, v_a^{(k)}\in S(w_a).
    \end{align}
\end{lemma}
\begin{proof}
We first calculate the fine virtual value for $v_a^{(k)}$
\begin{align}
    \Phi_{ij}^{\lambda^{\mathrm{F}}}(v_a^{(k)}) = v_{a,ij}^{(k)} - \frac{1}{f^{\mathrm{F}}(v_a^{(k)})} \sum_{w_b} \lambda_{i}^{\mathrm{F},v_{-i}}(v_{b,i}^{(k)}, v_{a,i}^{(k)}) \left( v_{b,ij}^{(k)} - v_{a,ij}^{(k)} \right).
\end{align}

Substitute $f^{\mathrm{F}} = \frac{1}{K^{mn}} f^{\mathrm{C}}$ and $\lambda^{\mathrm{F},v_{-i}} = \frac{1}{K^{mn}} \lambda^{\mathrm{C},w_{-i}}$
\begin{align}
    \Phi_{ij}^{\lambda^{\mathrm{F}}}(v_a^{(k)}) = v_{a,ij}^{(k)} - \frac{1}{f^{\mathrm{C}}(w_a)} \sum_{w_b} \lambda_{i}^{\mathrm{C}}(w_{b,i}, w_{a,i}) \left( v_{b,ij}^{(k)} - v_{a,ij}^{(k)} \right).    
\end{align}

Now, use the constant-gap property: $v_{b}^{(k)} = w_{b} - \delta^{(k)}$ and $v_{a}^{(k)} = w_{a} - \delta^{(k)}$. The difference is
\begin{align}
    v_{b,ij}^{(k)} - v_{a,ij}^{(k)} = w_{b,ij} - w_{a,ij}
\end{align}

The "jump size" is identical to the coarse jump. Plugging this back in, we have
\begin{align}
    \Phi_{ij}^{\lambda^{\mathrm{F}}}(v_a^{(k)}) & = (w_{a,ij} - \delta_{ij}^{(k)}) - \underbrace{\frac{1}{f^{\mathrm{C}}(w_a)} \sum \lambda_{i}^{\mathrm{C},w_{-i}} (w_{b,ij} - w_{a,ij})}_{\text{Coarse Penalty }}  = \Phi_{ij}^{\lambda^{\mathrm{C}}}(w_a) - \delta_{ij}^{(k)}.
\end{align}
Since $\delta_{ij}^{(k)} \ge 0$ for all $k$, it follows that:
\begin{align}
    \Phi_{ij}^{\lambda^{\mathrm{F}}}(v) \le \Phi_{ij}^{\lambda^{\mathrm{C}}}(R(v)).\label{eq:virt-dom}
\end{align}
\end{proof}

\begin{theorem}[Revenue Upper Bound for the Uniform Distribution]\label{lem:bound_uniform}
    With $\lambda^{\mathrm{F},v_{-i}}$ defined in Eq.~\ref{eq:lift-type_uniform} and Eq.~\ref{eq:lift-sink_uniform},
    \begin{align}
        \Rev\!\left(D^{\mathrm{F}}\right) \ \le\ \Rev\!\left(D^{\mathrm{C}}\right).
    \end{align}
\end{theorem}

By Eq.~\ref{eq:virt-dom},
\[
\sum_{j=1}^m \max\{0,\max_i \Phi_{ij}^{\lambda^{\mathrm{F}}}(v)\}
\ \le\
\sum_{j=1}^m \max\{0,\max_i \Phi_{ij}^{\lambda^{\mathrm{C}}}(R(v))\},
\qquad \forall v\in\supp(D^{\mathrm{F}}).
\]
Taking expectation over $v\sim D^{\mathrm{F}}$ and using the pushforward identity Eq.~\ref{eq:push_n},
\begin{align*}
&\E_{v\sim D^{\mathrm{F}}}\!\left[\sum_{j} \max\{0,\max_i \Phi_{ij}^{\lambda^{\mathrm{F}}}(v)\}\right]
\le
\E_{v\sim D^{\mathrm{F}}}\!\left[\sum_{j} \max\{0,\max_i \Phi_{ij}^{\lambda^{\mathrm{C}}}(R(v))\}\right] \\
=& 
\E_{w\sim D^{\mathrm{C}}}\mathbb{E}_{v\in S(w)}\!\left[\sum_{j} \max\{0,\max_i \Phi_{ij}^{\lambda^{\mathrm{C}}}(w)\}\right]
= \E_{w\sim D^{\mathrm{C}}}\!\left[\sum_{j} \max\{0,\max_i \Phi_{ij}^{\lambda^{\mathrm{C}}}(w)\}\right].
\end{align*}
The LHS is the dual value of a feasible dual solution for the fine distribution, and the RHS is the dual value of the corresponding coarse distribution solution.

Hence,
\begin{align}
    \min_{\lambda^{\mathrm{F}}\ \text{flow}} \E_{v\sim D^{\mathrm{F}}}\!\left[\sum_{j}
    \max\{0,\max_i \Phi_{ij}^{\lambda^{\mathrm{F}}}(v)\}\right]
    \ \le\
    \min_{\lambda^{\mathrm{C}}\ \text{flow}} \E_{w\sim D^{\mathrm{C}}}\!\left[\sum_{j}
    \max\{0,\max_i \Phi_{ij}^{\lambda^{\mathrm{C}}}(w)\}\right],    \label{eq:uniform_domi_raw}
\end{align}
which is exactly $\Rev\!\left(D^{\mathrm{F}}\right) \ \le\ \Rev\!\left(D^{\mathrm{C}}\right)$ by the dual representation \cref{eq:dual} for each of the discrete problems. \cref{eq:uniform_domi_raw} follows by constructing a fine-grid dual variable $\lambda^\mathrm{F}$ from the coarse-grid optimal $\lambda^\mathrm{C}$ (RHS). The lifted $\lambda^\mathrm{F}$ is feasible for the fine problem, and its dual objective value upper bounds the optimal dual objective of the fine problem (LHS).
More specifically,
\begin{align}
    \Rev\!\left(D^{\mathrm{F}}\right) \le \E_{v\sim D^{\mathrm{F}}}\!\left[\sum_{j} \max\{0,\max_i \Phi_{ij}^{\lambda^{\mathrm{F}}}(v)\}\right] 
    = \E_{w\sim D^{\mathrm{C}}}\E_{v\sim S(w)}\!\left[\sum_{j} \max\{0,\max_i \left(\Phi_{ij}^{\lambda^{\mathrm{C}}}(w)-(w_{ij}-v_{ij})\right)\}\right].\nonumber
\end{align}

\section{Upper Bound on Continuous Revenue}\label{sec:method:continuous}

Having established the properties of optimal revenue for dual solutions computed on discrete  grids, we now generalize the lifting procedure to the continuous domain, establishing that the revenue bounds apply to the original, continuous distribution.  For uniform distributions, we obtain a direct, certified upper bound on the continuous revenue from any coarse-grid solution. For general distributions, we establish that the discretized bound converges to a valid continuous bound as the grid is refined, with an explicit error term that vanishes with the mesh size.

\subsection{Model the Lagrange Multipliers as Kernels}

In the discrete problem on a grid $\mathcal{G}$, dual variables $\lambda_i^{\mathcal{G}, t_{-i}}(t_i,t_i')$, where $(t_i,t_{-i})\in\mathcal{G}$, $(t_i',t_{-i})\in\mathcal{G}$ are Lagrange multipliers attached to each DSIC constraint.  These dual variables can be learned by deep neural networks for general distributions, and can be further lifted for the uniform distribution (\cref{sec:method:lift}). In continuous type spaces, we model them as \emph{kernels} (or {\em measures}). Going forward, we use $t\in \mathcal{G}$ to denote a discretized type and $v_i$ to denote a continuous type, for clarity. The kernel formulation in this subsection applies to general value distributions.
\begin{definition}[Dual Kernel]
Fix $i\in[n]$. For almost every $v_{-i}$, let
\[
\lambda_i^{v_{-i}}:\ V_i\times \mathcal{B}(V_i\cup\{\varnothing\})\to [0,\infty]
\]
be a nonnegative kernel, so that for each $v_i$, we have 
$A\mapsto \lambda_i^{v_{-i}}(v_i,A)$ is a finite measure on
$V_i\cup\{\varnothing\}$, and for each measurable $A$, we have
$v_i\mapsto \lambda_i^{v_{-i}}(v_i,A)$ is measurable.    
\end{definition}

We say $\lambda$ \emph{conserves flow} if for every bidder $i$ and almost every
$v_{-i}$, for almost every $v_i$,
\begin{equation}\label{eq:flow-cont}
\lambda_i^{v_{-i}}(v_i,V\cup\{\varnothing\})
\;=\;
f(v_i,v_{-i}) \;+\; \int_{V_i}\lambda_i^{v_{-i}}(u_i,\mathrm{d}v_i) \mathrm{d}u_i.
\end{equation}

The expected revenue is
\[
\Rev^{(x,p)}(D):=\int_V f(v)\Big(\sum_{i=1}^n p_i(v)\Big) \dV,
\qquad
\Rev(D):=\sup_{\text{DSIC,IR,feasible}\ (x,p)}\Rev^{(x,p)}(D).
\]

\subsection{Lifting Discrete Solutions to a Continuous Flow-Conserving Kernel}

We now build a continuous, feasible $\lambda$ from a discrete feasible $\lambda^\mathcal{G}$ on grid $\mathcal{G}$ by \emph{spreading}
each coarse flow uniformly inside each cell in the discrete grid with respect to the
\emph{conditional density}. The results in this subsection apply to general value distributions.
For bidder $i$ define, for almost every $v_i$, the conditional density 
\(
s_i(v_i)
:=
\frac{f_i(v_i)}{f^\mathcal{G}_i(R(v_i))}.
\)
Then for every cell, $C_i(t_i)$, 
\[
\int_{C_i(t_i)} s_i(v_i) \mathrm{d}v_i
=
\frac{1}{f^\mathcal{G}_i(t_i)}\int_{C_i(t_i)} f_i(v_i) \mathrm{d}v_i
=1, \qquad \forall t_i\in\mathcal{G}.
\]

Similarly, for $v_{-i}$ define
\[
s_{-i}(v_{-i})
:=
\frac{f_{-i}(v_{-i})}{f^\mathcal{G}_{-i}(R(v_{-i}))},
\qquad
f^\mathcal{G}_{-i}(t_{-i})=\prod_{k\ne i}f^\mathcal{G}_k(t_k).
\]
Then $\int_{C_{-i}(t_{-i})} s_{-i}(v_{-i}) \mathrm{d}v_{-i}=1$ for each cell of
others.

\begin{definition}[Lifted Kernel]
    Fix bidder $i$. For almost every $v_{-i}$, let $t_{-i}:=R(v_{-i})$.
For $v_i\in V_i$ and a measurable set $A\subseteq V_i$, define
\begin{equation}\label{eq:lift-kernel}
\lambda_i^{v_{-i}}(v_i,A)
:=
s_{-i}(v_{-i}) s_i(v_i) 
\sum_{t_i'\in\mathcal{G}}
\lambda^{\mathcal{G},t_{-i}}_i(R(v_i),t_i') 
\int_{A\cap C_i(t_i')} s_i(u_i) \mathrm{d}u_i,
\end{equation}
and define the mass to the sink by
\begin{equation}\label{eq:lift-kernel-sink}
\lambda_i^{v_{-i}}(v_i,\{\varnothing\})
:=
s_{-i}(v_{-i}) s_i(v_i) 
\lambda^{\mathcal{G},t_{-i}}_i(R(v_i),\varnothing).
\end{equation}
\end{definition}

\begin{restatable}[Lifted $\lambda$ conserves flow]{lemma}{lcf}
    \label{lem:lift-flow}
    Assume $\lambda^\mathcal{G}$ satisfies the discrete flow conservation. Then the lifted kernel $\lambda$ defined by
\cref{eq:lift-kernel}--\cref{eq:lift-kernel-sink} satisfies the continuous
flow conservation \cref{eq:flow-cont} for all bidders.
\end{restatable}
We defer the proof to \cref{appx:proof}.

\noindent\textbf{Virtual Values}.\ \ 
Given any flow-conserving $\lambda$, define (for almost every $v$ with $f(v)>0$)
\begin{equation}\label{eq:virt-cont}
\Phi_{ij}^{\lambda}(v)
\;:=\;
v_{ij}
\;-\;
\frac{1}{f(v)}
\int_{V_i} (u_{ij}-v_{ij}) \lambda_i^{v_{-i}}(u_i,\mathrm{d}v_i)\mathrm{d}u_i.
\end{equation}

\begin{restatable}[Continuous Weak Duality]{lemma}{cwd}\label{lem:cont-thm24}
Assume $\lambda$ conserves flow in the sense of \cref{eq:flow-cont}.
Then for every DSIC/IR feasible mechanism $(x,p)$,
\begin{equation}\label{eq:cont-thm24}
\int_V f(v)\Big(\sum_{i=1}^n p_i(v)\Big) \dV
\;\le\;
\int_V f(v)\Big(\sum_{i=1}^n\sum_{j=1}^m x_{ij}(v) \Phi_{ij}^\lambda(v)\Big) \dV.
\end{equation}
Consequently,
\[
\Rev(D)\;\le\;
\sup_{x\ \text{feasible}}\int_V f(v)\Big(\sum_{i,j}x_{ij}(v)\Phi_{ij}^\lambda(v)\Big) \dV
\;=\;
\int_V f(v)\Big(\sum_{j=1}^m \max_{i\in[n]}(\Phi_{ij}^\lambda(v))_+\Big) \dV.
\]
\end{restatable}

\begin{proof}
DSIC says that for every $i$, almost every $v_{-i}$, 
every $v_i\in V_i$, and every report $v_i'\in V_i\cup\{\varnothing\}$,
\[
\sum_{j=1}^m x_{ij}(v_i,v_{-i}) v_{ij}-p_i(v_i,v_{-i}) \;\ge\; \sum_{j=1}^m x_{ij}(v_i',v_{-i}) v_{ij}-p_i(v_i',v_{-i}).
\]
Rearranging,
\begin{equation}\label{eq:DSIC-rearr}
\sum_{j=1}^m v_{ij}\big(x_{ij}(v_i,v_{-i})-x_{ij}(v_i',v_{-i})\big)
\;-\;
\big(p_i(v_i,v_{-i})-p_i(v_i',v_{-i})\big)
\;\ge\;0.
\end{equation}
Multiply \cref{eq:DSIC-rearr} by the nonnegative measure
$\lambda_i^{v_{-i}}(v_i,\mathrm{d}v_i')$ and integrate on $v_i'$.
Since the integrand is nonnegative, Tonelli's theorem applies and yields, for
each $i$ and almost every $v_{-i}$ and every $v_i$,
\begin{align}
0
&\le
\int_{V_i\cup\{\varnothing\}}
\Bigg[
\sum_{j=1}^m v_{ij}\big(x_{ij}(v_i,v_{-i})-x_{ij}(v_i',v_{-i})\big)
-\big(p_i(v_i,v_{-i})-p_i(v_i',v_{-i})\big)
\Bigg] 
\lambda_i^{v_{-i}}(v_i,\mathrm{d}v_i').
\label{eq:integrated-dsic}
\end{align}
Now integrate \cref{eq:integrated-dsic} over $v_i$ and $v_{-i}$:
\begin{align}
0
&\le
\int_{V_{-i}}\int_{V_i}
\int_{V_i\cup\{\varnothing\}}
\Bigg[
\sum_{j=1}^m v_{ij}\big(x_{ij}(v_i,v_{-i})-x_{ij}(v_i',v_{-i})\big)
-\big(p_i(v_i,v_{-i})-p_i(v_i',v_{-i})\big)
\Bigg] 
\lambda_i^{v_{-i}}(v_i,\mathrm{d}v_i') \mathrm{d}v_i \mathrm{d}v_{-i}.
\label{eq:triple}
\end{align}

Add the revenue to both sides, i.e., add
$\int_V f(v)p_i(v) \dV$ (and later sum over $i$).
Define
\[
\mathcal{R}_i := \int_V f(v)p_i(v) \dV.
\]
Then \cref{eq:triple} implies
\begin{align}
\mathcal{R}_i
&\le
\mathcal{R}_i
+
\int_{V_{-i}}\int_{V_i}
\int_{V_i\cup\{\varnothing\}}
\Bigg[
\sum_{j=1}^m v_{ij}\big(x_{ij}(v)-x_{ij}(v_i',v_{-i})\big)
-\big(p_i(v)-p_i(v_i',v_{-i})\big)
\Bigg] 
\lambda_i^{v_{-i}}(v_i,\mathrm{d}v_i') \mathrm{d}v_i \mathrm{d}v_{-i}.
\label{eq:addrev}
\end{align}
We now expand the right-hand side and collect the payment terms. 
Because all terms are integrable on the bounded support and $\lambda$ is finite by assumption, Tonelli's theorem lets us regroup the iterated integrals and collect, for almost every $v$, the coefficient of
$p_i(v)$, which is
\[
f(v)
+
\int_{V_i}\lambda_i^{v_{-i}}(u_i,\mathrm{d}v_i)\mathrm{d}u_i
-
\lambda_i^{v_{-i}}(v_i,V_i\cup\{\varnothing\}).
\]
By flow conservation \cref{eq:flow-cont}, this coefficient equals $0$ for a.e.
$v$. Hence all payment terms cancel, and \cref{eq:addrev} reduces to
\begin{align}
\mathcal{R}_i
&\le
\int_{V_{-i}}\int_{V_i}
\sum_{j=1}^m x_{ij}(v) 
\Bigg[
v_{ij}
-
\frac{1}{f(v)}
\int_{V_i}(u_{ij}-v_{ij}) \lambda_i^{v_{-i}}(u_i,\mathrm{d}v_i)\,\mathrm{d}u_i
\Bigg]
f(v) \mathrm{d}v_i \mathrm{d}v_{-i}
\nonumber\\
&=
\int_V f(v)\sum_{j=1}^m x_{ij}(v) \Phi_{ij}^{\lambda}(v) \dV,
\label{eq:rev-to-virt}
\end{align}
where $\Phi^\lambda$ is exactly \cref{eq:virt-cont}.
Summing \cref{eq:rev-to-virt} over $i$ gives \cref{eq:cont-thm24}.

Finally, for each fixed profile $v$, the maximization over feasible $x(\cdot)$
of $\sum_{i,j}x_{ij}(v)\Phi_{ij}^\lambda(v)$ subject to
$x_{ij}(v)\ge 0$ and $\sum_i x_{ij}(v)\le 1$ is separable across items $j$ and
equals $\sum_j \max_i(\Phi_{ij}^\lambda(v))_+$.
This
yields the final displayed bound, with a detailed derivation in \cref{appx:proof}.

\end{proof}

\cref{lem:cont-thm24} says  any flow-conserving $\lambda$ is a
dual-feasible certificate for the continuous problem, and its induced virtual
values yield an explicit revenue upper bound.  This result holds for any continuous distribution, not only the uniform case.

\subsection{Upper Bound on Continuous Revenue under Uniform Distribution}\label{sec:method:cub}

For uniform value distributions, we combine the discrete lifting result (\cref{lem:bound_uniform}), which shows that finer grids yield tighter bounds, with continuous weak duality (\cref{lem:cont-thm24}) to obtain a certified upper bound on the continuous distribution optimal revenue directly from any coarse-grid dual solution.
\begin{restatable}[Upper Bound on Continuous Revenue under Uniform Distribution]{theorem}{ubu}
When the distribution is uniform, the upper bound on the revenue of a DSIC, multi-bidder, multi-item auction is:
    \begin{align}
\Rev\!\left(D\right) 
\le & \E_{w\sim D^{\mathrm{C}}}\E_{v\sim C(w)}\!\left[\sum_{j} \max\{0,\max_i \left(\Phi_{ij}^{\lambda^{\mathrm{C}}}(w)-(w_{ij}-v_{ij})\right)\}\right],
\end{align}
where $\lambda^{\mathrm{C}}$ is a solution for a discretized problem.
\label{Thm:ct}
\end{restatable}

The proof can be found in \cref{appx:proof}. Two features of \cref{Thm:ct} are worth highlighting. First, the bound holds for  any coarse grid $\mathcal{G}^\mathrm{C}$ and any flow-conserving $\lambda^\mathrm{C}$, including $\lambda^\mathrm{C}$ that are suboptimal for the discrete problem on $\mathcal{G}^\mathrm{C}$. This is a property we exploit in practice, since our solver produces flow-conserving duals by deep learning. Second, although the statement involves only coarse-grid virtual values $\Phi^{\lambda^\mathrm{C}}$ and a within-cell offset $w_{ij}-v_{ij}$, lifting is essential to the proof: it constructs (via \cref{eq:lift-kernel,eq:lift-kernel-sink}) a continuous flow-conserving kernel $\lambda$ whose induced virtual values satisfy the pointwise identity
\[
\Phi^\lambda_{ij}(v)\;=\;\Phi^{\lambda^\mathrm{C}}_{ij}(R(v))\;-\;\bigl(R(v)_{ij}-v_{ij}\bigr),
\]
and continuous weak duality (\cref{lem:cont-thm24}) then turns this identity into the stated revenue bound.

\subsection{Upper Bound on Continuous Revenue under a General Distribution}

In \cref{sec:method:cub}, we discuss how to compute a valid upper bound on optimal revenue for a uniform value distribution. 
Unlike the uniform case, where any coarse-grid flow-conserving dual solution already yields a valid bound on the continuous revenue (\cref{Thm:ct}), for general distributions 
we show that the discretized bound plus a computable correction term provides a valid continuous certificate, with the correction vanishing as the grid is refined.

Denote the coarse virtual values by $\Phi^{ \lambda^{\mathcal{G}}}$. For computational purposes, one can write the virtual value under the original, continuous distribution $\Phi^\lambda$ in terms of the grid
data and \emph{cell conditional means}. Define, for each bidder $i$, item $j$,
and each grid point $t_i\in\mathcal G_i$,
\[
\zeta_{ij}(t_i)
:=
\E[v_{ij}\mid R_i(v_i)=t_i]
=
\int_{C_i(t_i)} u_{ij}\,s_i(u_i)\,\mathrm{d}u_i.
\]
Then a direct substitution of \cref{eq:lift-kernel} into \cref{eq:virt-cont}
gives, for almost every $v$ with $t=R(v)$,
\begin{equation}\label{eq:Phi-lifted-explicit}
\Phi_{ij}^{\lambda}(v)
=
v_{ij}
-
\frac{1}{f^\mathcal{G}(t)}
\sum_{t_i'\in\mathcal{G}_i}
\lambda^{\mathcal{G},t_{-i}}_i(t_i',t_i)\,\big(\zeta_{ij}(t_i')-v_{ij}\big).
\end{equation}
To see this, the factors $s_i,s_{-i}$ cancel exactly against $f(v)$. We then compare the lifted
$\Phi^\lambda(v)$ from \cref{eq:Phi-lifted-explicit} to
$\Phi^{ \lambda^{\mathcal{G}}}(t)$ for almost every $v$ and $t= R(v)$.

\begin{lemma}\label{lem:one-sided}
For almost every $v\in V$ and $t= R(v)$, for every bidder $i$ and item $j$,
\begin{equation}\label{eq:one-sided}
\Phi_{ij}^{\lambda}(v)
\;\le\;
\Phi^{ \lambda^{\mathcal{G}}}_{ij}(t)
\;+\;
\Delta \alpha_i(t),
\end{equation}
where $\Delta$ is the mesh size and $\alpha_i(t)=\frac{1}{f^\mathcal{G}(t)}
\sum_{t_i'}\lambda^{\mathcal{G},t_{-i}}_i(t_i',t_i)>0$ is the normalized incoming-flow ratio.
\end{lemma}

\begin{proof}
Fix $i,j$ and $v$ with $t= R(v)$.
Start from the explicit lifted formula \cref{eq:Phi-lifted-explicit}, and
expand the term in parentheses:
\[
\zeta_{ij}(t_i')-v_{ij} = (\zeta_{ij}(t_i')-t_{ij}') + (t_{ij}'-t_{ij}) + (t_{ij}-v_{ij}).
\]
Substitute and regroup:
\begin{align*}
\Phi_{ij}^{\lambda}(v)
&=
v_{ij}
-
\frac{1}{f^\mathcal{G}(t)}
\sum_{t_i'}
\lambda^{\mathcal{G},t_{-i}}_i(t_i',t_i) (t_{ij}'-t_{ij})
\\
&\quad
-\frac{1}{f^\mathcal{G}(t)}
\sum_{t_i'}
\lambda^{\mathcal{G},t_{-i}}_i(t_i',t_i) (\zeta_{ij}(t_i')-t_{ij}')
\\
&\quad
-\frac{1}{f^\mathcal{G}(t)}
\sum_{t_i'}
\lambda^{\mathcal{G},t_{-i}}_i(t_i',t_i) (t_{ij}-v_{ij}).
\end{align*}

For the first line, since
\(
\Phi^{ \lambda^{\mathcal{G}}}_{ij}(t)
=
t_{ij}
-
\frac{1}{f^\mathcal{G}(t)}
\sum_{t_i'}\lambda^{\mathcal{G},t_{-i}}_i(t_i',t_i) (t_{ij}'-t_{ij}),
\)
replacing $t_{ij}$ by $v_{ij}$ changes it by $v_{ij}-t_{ij}\le 0$ (since $t_{ij}$
is the grid point dominating the cell containing $v_{ij}$).
Formally, write
\[
v_{ij}
-
\frac{1}{f^\mathcal{G}(t)}\sum_{t_i'}\lambda^{\mathcal{G},t_{-i}}_i(t_i',t_i)(t_{ij}'-t_{ij})
=
\Phi^{ \lambda^{\mathcal{G}}}_{ij}(t)\;+\;(v_{ij}-t_{ij})
\;\le\;
\Phi^{ \lambda^{\mathcal{G}}}_{ij}(t).
\]

The second line 
\[
-\frac{1}{f^\mathcal{G}(t)}
\sum_{t_i'}\lambda^{\mathcal{G},t_{-i}}_i(t_i',t_i) (\zeta_{ij}(t_i')-t_{ij}')
=
\frac{1}{f^\mathcal{G}(t)}
\sum_{t_i'}\lambda^{\mathcal{G},t_{-i}}_i(t_i',t_i) (t_{ij}'-\zeta_{ij}(t_i'))
\;\le\;
\Delta \alpha_i(t).
\]

The third line is \emph{nonpositive}, since $t_{ij}-v_{ij}\ge 0$ and
$\lambda^{\mathcal{G}}\ge 0$. Combining these three bullets yields exactly \cref{eq:one-sided}.
\end{proof}

\begin{theorem}[Upper Bound on Continuous Revenue under General Distribution]\label{thm:coarse-plus-error}
Let $\lambda^{\mathcal{G}}$ be any coarse flow-conserving dual solution for grid
$\mathcal{G}$, and let $\mathrm{UB}_{\mathrm{coarse}}$ be its (known) coarse dual
objective value.
Then the continuous optimal revenue satisfies
\begin{equation}\label{eq:final-bound}
\Rev(D)
\;\le\;
\mathrm{UB}_{\mathrm{coarse}}
\;+\;
m \Delta\cdot
\E_{t\sim\mathcal{G}}\Big[\max_{i\in[n]}\alpha_i(t)\Big],
\end{equation}
where $\Delta$ is the width of a grid cell in $\mathcal{G}$, $\alpha_i(t)$ is the
normalized incoming-flow ratio (\cref{lem:one-sided}), and the correction term $m \Delta\cdot
\E_{t\sim\mathcal{G}}\Big[\max_{i\in[n]}\alpha_i(t)\Big]$ converges to zero as $\Delta\rightarrow 0$.
\end{theorem}

\begin{proof}
Since the lifted kernel $\lambda$ is a feasible dual solution, the optimal revenue on the continuous problem must be smaller than the lifted dual objective:
\[
\Rev(D)\le \int_V f(v)\Big(\sum_{j=1}^m \max_i(\Phi_{ij}^{\lambda}(v))_+\Big) \dV.
\]
Fix $v$ and let $t= R(v)$. Apply Lemma~\ref{lem:one-sided}:
\[
\Phi_{ij}^{\lambda}(v)\le \Phi^{ \lambda^{\mathcal{G}}}_{ij}(t)+\Delta\alpha_i(t).
\]
Since $\Delta\alpha_i(t)\ge 0$, we have
\[
\max_i(\Phi_{ij}^{\lambda}(v))_+
\le
\max_i\big(\Phi^{ \lambda^{\mathcal{G}}}_{ij}(t)+\Delta\alpha_i(t)\big)_+
\le
\max_i\big(\Phi^{ \lambda^{\mathcal{G}}}_{ij}(t)\big)_+
+
\Delta\max_i\alpha_i(t).
\]
Sum over $j$:
\[
\sum_{j=1}^m \max_i(\Phi_{ij}^{\lambda}(v))_+
\le
\sum_{j=1}^m \max_i\big(\Phi^{ \lambda^{\mathcal{G}}}_{ij}(t)\big)_+
+
m\Delta\max_i\alpha_i(t).
\]
Now take expectation over $v\sim D$. The first term of r.h.s. depends on $v$ only through
$t= R(v)$, so its expectation equals the expectation under $t\sim\mathcal{G}$:
\[
\E_{v\sim D}\Big[\sum_{j=1}^m \max_i(\Phi^{ \lambda^{\mathcal{G}}}_{ij}( R(v)))_+\Big]
=
\E_{t\sim\mathcal{G}}\Big[\sum_{j=1}^m \max_i(\Phi^{ \lambda^{\mathcal{G}}}_{ij}(t))_+\Big]
=
\mathrm{UB}_{\mathrm{coarse}}.
\]
The second term likewise becomes
\[
m\Delta\E_{v\sim D}\big[\max_i\alpha_i( R(v))\big]
=
m\Delta\E_{t\sim\mathcal{G}}\big[\max_i\alpha_i(t)\big].
\]
Combining yields \cref{eq:final-bound}. 

\paragraph{Boundedness of the correction term under the absorbing Markov chain
parameterization.} We first show that,
when $\lambda^{\mathcal{G}}$ is produced by the parameterization of
\cref{sec:method:flow} with leakage floor $\varepsilon>0$, the expectation
$\E_{t\sim\mathcal{G}}[\max_i\alpha_i(t)]$ is bounded uniformly.

Fix bidder $i$ and any context $v_{-i}$. Each row of $\Pi_i^{v_{-i}}$ sums to
at most $1-\varepsilon$, and the source vector $\nu^{v_{-i}}$ has entries
$f^\mathcal{G}(u,v_{-i})=f^\mathcal{G}_i(u)\,f^\mathcal{G}_{-i}(v_{-i})$ by
bidder independence, so
$\sum_{u\in\mathcal{G}_i}\nu^{v_{-i}}(u)=f^\mathcal{G}_{-i}(v_{-i})$.

Summing the flow equation $\mu=\nu+(\Pi_i^{v_{-i}})^\top\mu$ entrywise yields
\[
\sum_{u\in\mathcal{G}_i}\mu(u)
\;=\;
f^\mathcal{G}_{-i}(v_{-i})
+
\sum_{u\in\mathcal{G}_i}\mu(u)\bigl(1-\alpha_i^{v_{-i}}(u)\bigr)
\;\le\;
f^\mathcal{G}_{-i}(v_{-i})
+
(1-\varepsilon)\sum_{u\in\mathcal{G}_i}\mu(u),
\]
hence $\sum_{u}\mu(u)\le f^\mathcal{G}_{-i}(v_{-i})/\varepsilon$. The total
type-to-type inflow at context $v_{-i}$ is therefore
\[
\sum_{t_i,t_i'\in\mathcal{G}_i}\lambda_i^{\mathcal{G},v_{-i}}(t_i',t_i)
\;=\;
\sum_{u\in\mathcal{G}_i}\mu(u)\bigl(1-\alpha_i^{v_{-i}}(u)\bigr)
\;\le\;
\frac{1-\varepsilon}{\varepsilon}\,f^\mathcal{G}_{-i}(v_{-i}).
\]
Averaging over $t\sim\mathcal{G}$, the factor $f^\mathcal{G}(t)$ in the
denominator of $\alpha_i(t)$ cancels:
\[
\E_{t\sim\mathcal{G}}\bigl[\alpha_i(t)\bigr]
\;=\;
\sum_{t_{-i}}\sum_{t_i,t_i'}\lambda_i^{\mathcal{G},t_{-i}}(t_i',t_i)
\;\le\;
\frac{1-\varepsilon}{\varepsilon}.
\]
Since $\alpha_i\ge 0$, then $\E[\max_i\alpha_i(t)]\le\sum_i\E[\alpha_i(t)]
\le n(1-\varepsilon)/\varepsilon$, which depends only on $\varepsilon$ and not
on the grid. Consequently, the correction term in \cref{eq:final-bound} is
bounded by $mn\Delta(1-\varepsilon)/\varepsilon$, hence $O(\Delta)$ at fixed
$\varepsilon$ and vanishing as $\Delta\to 0$.

\end{proof}

\cref{thm:coarse-plus-error} allows our deep learning based, dual computation framework to be applied to general distributions, while still yielding a valid bound on the optimal continuous revenue.
The theorem confirms that as the grid is refined and $\Delta \to 0$, the correction term vanishes, and the discretized dual objective converges to a valid upper bound on the continuous optimal revenue.

In principle, \cref{thm:coarse-plus-error} can also be applied to compute a valid upper bound for any value distribution (including Beta, for example) at any finite grid resolution, by adding the explicit correction term $m\Delta \cdot \E[\max_i \alpha_i(t)]$. In practice, this correction is loose for coarse grids, motivating the use of finer discretizations.

\section{Experiments}
\label{sec:experiments}

We evaluate our approach on a hierarchy of different auction problems, moving from analytically solvable discrete type distribution problems to high-dimensional continuous type distribution problems. 
For discrete type domains, we validate our neural, dual problem solver by comparing its computed dual objective against the true optimal revenue obtained via linear programming (LP) to solve the primal auction design problem. 

For continuous type domains in which the true optimum is unknown, we assess the tightness of our results by comparing our \textbf{Certified Upper Bound} (\cref{Thm:ct}) against the best achievable revenue from \textbf{Primal Lower Bounds} (e.g., GemNet). A small difference between these bounds 
certifies that the auction computed from the primal method is near-optimal. To compute these bounds, we first learn a feasible dual solution on a discretized grid, 
 and then compute the {\bf Certified Lifted Bound} by applying the lifting map of \cref{Thm:ct} together with exact integration over the within-cell offsets $(w_{ij}-v_{ij})$. For uniform valuations this lifted value is exactly the {\bf Certified Upper Bound} on the continuous revenue guaranteed by \cref{Thm:ct}; we use the two terms interchangeably, with ``lifted'' used to emphasize the computational procedure.
For general distributions, where a lifting map is not yet established, we report the {\bf Exact Discrete Bound}, which is shown in \cref{thm:coarse-plus-error} to converge to the continuous bound as the discretization resolution $K$ approaches infinity. At any finite grid resolution, the {\bf Exact Discrete Bound} is a valid upper bound for the discretized problem; by \cref{thm:coarse-plus-error}, this bound converges to a valid upper bound on the continuous revenue as the grid is refined.
We provide the detailed architectures and hyperparameters of the deep learning
methods in \cref{app:implementation}.

\subsection{Validation: Exact Recovery in the Binary Valuation Setting}
We first verify the correctness of our flow-based parameterization using the auction studied by  \citet{yao2017dominant}, a canonical instance with two items, $n$ bidders, and binary valuations $\{a, b\}$, with $0<a<b$. As the domain is discrete, this setting evaluates the expressiveness of the neural network and effectiveness of optimization without concern to discretization error. We define the problem instance with parameters $a=3$ and $p=0.3$ (the probability of the low type), varying $b$. We emphasize this instance because, as noted in the related work, it represents the only known analytical characterization for a multi-bidder, multi-item DSIC auction. Recovering these values certifies that our solver can capture the complex interaction between bidders in the optimal dual solution.

\Cref{tab:yao} compares our computed bounds against the known analytical optimal revenue. In all configurations, 
for different values of $b$ and $n$, our method recovers the theoretical optimum with high precision. This confirms that the absorbing Markov chain parameterization is sufficiently expressive to represent the binding constraints of the optimal mechanism. Crucially, because our architecture enforces constraints by construction, the optimality gap is driven solely by the convergence of the optimizer rather than feasibility violations.

\begin{table}[h]
    \caption{{\bf Certified upper bounds} for the multi-bidder, 2-item auctions in \citet{yao2017dominant} ($n\times $2, discrete types, $a=3, p=0.3$). Our framework recovers the analytical optimal revenue, demonstrating the exactness of the neural network dual solver.}
    \label{tab:yao}
    \centering
    \small
    \begin{tabular}{l c c c c | c c c}
        \toprule
        \multirow{2}{*}{Method} & \multicolumn{4}{c}{\textbf{Varying High Valuation} ($n=2$)} & \multicolumn{3}{c}{\textbf{Varying Bidders} ($b=4$)} \\
        \cmidrule(lr){2-5} \cmidrule(lr){6-8}
         & $b=7$ & $b=5$ & $b=4$ & $b=3.5$ & $n=3$ & $n=5$ & $n=8$ \\
        \midrule
        Optimal (Analytical) & 12.7400 & 9.1504 & 7.4774 & 6.7221 & 7.8309 & 7.9840 & 7.9996 \\
        \textbf{Ours (Certified Bound)} & \textbf{12.7400} & \textbf{9.1504} & \textbf{7.4774} & \textbf{6.7221} & \textbf{7.8309} & \textbf{7.9840} & \textbf{7.9996} \\
        \midrule
        GemNet & 12.7400 & 9.1504 & 7.4774 & 6.7221 & 7.8309 & 7.9830 & 7.9996 \\
        RegretNet  & 12.8052 & 9.1738 & 7.5017 & 6.6799 & 7.8370 & 7.9980 & 8.0000 \\
        (IC violation)  & 0.00951 & 0.00846 & 0.00548 & 0.0094 & 0.0383 & 0.0396 & 0.0363 \\
        \bottomrule
    \end{tabular}
\end{table}

\subsection{Qualitative Validation: Single-Bidder Auction}

We next examine a single-bidder, multi-item setting for which there is an optimal, analytical solution, both to validate the tightness of the bound that we compute and to understand whether the dual solver captures the structural properties of the optimal mechanism. For this, we consider a setting with a single bidder with two items and valuations drawn from $U[0,1]$. \citet{manelli2006bundling} characterize the optimal mechanism, which is known to exhibit complex bundling behavior defined by diagonal decision boundaries and has expected revenue of \textbf{0.549}.

\begin{wrapfigure}[17]{r}{0.4\textwidth}
  \centering
  \vspace{-40pt} 
  \includegraphics[width=\linewidth]{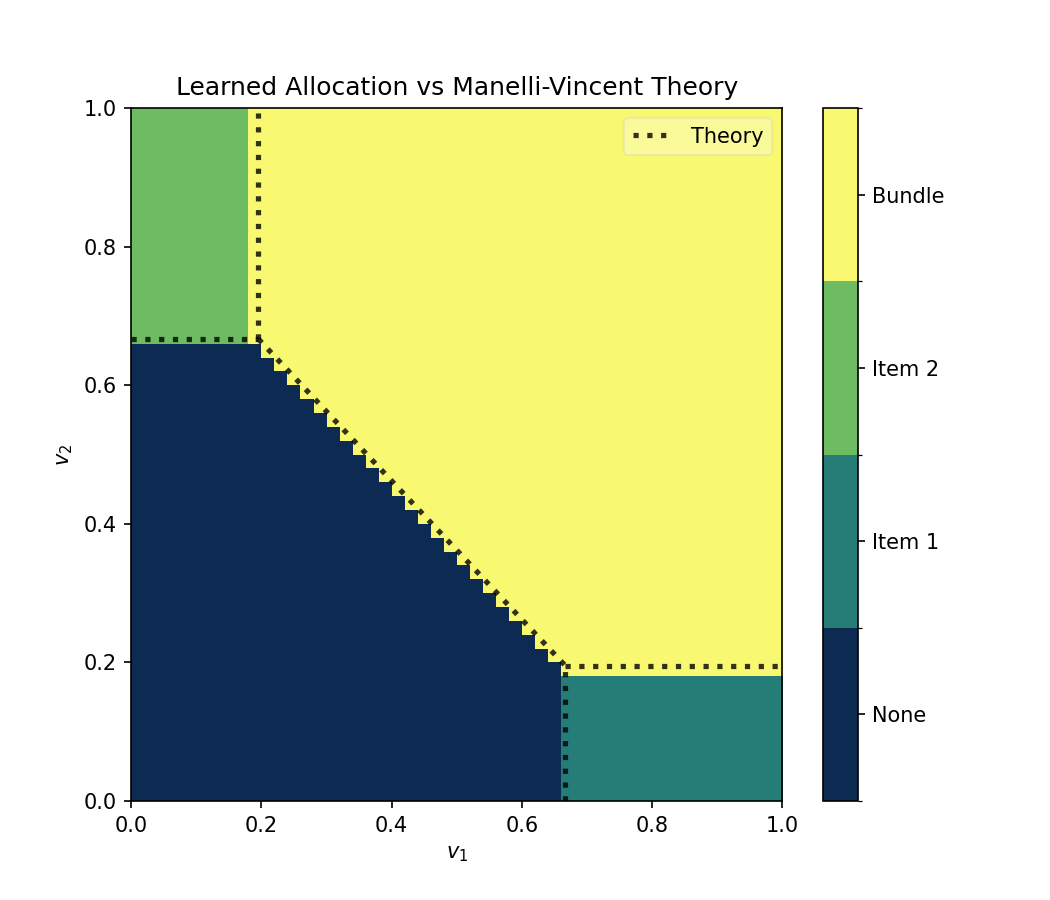}
  \caption{\textbf{Allocations induced by learned duals vs. theory (1$\times$2, Uniform).} The heatmap shows the allocation rule (\cref{eq:maxx}) specified by learned virtual values $\Phi_1, \Phi_2$. The dotted lines represent the analytical results derived by \citet{manelli2006bundling}. }
  \label{fig:alloc_viz}
\end{wrapfigure}

We trained the dual solver by discretizing the continuous valuation domain $[0,1]^2$ into a uniform $50 \times 50$ grid. We visualize the learned decision regions in \Cref{fig:alloc_viz}. Recalling \cref{eq:maxx}, the revenue-maximizing allocation for a fixed dual $\lambda$ assigns item $j$ if and only if the virtual value is positive ($\Phi_{j}^{\lambda}(v) > 0$). As shown, the network and the dual solution automatically discovers the diagonal decision boundaries characteristic of bundling mechanisms, recovering the correct structural form without any prior knowledge of the allocation rule.

On this $50 \times 50$ discretization, we compute an exact discrete bound of \textbf{0.569}. The lifting theory in \cref{Thm:ct} complements this learned bound by tightening the certificate beyond the discrete approximation. By integrating over the within-cell offsets $(w_{ij} - v_{ij})$ as specified by the lifting map, the lifting map tightens the certificate to \textbf{0.558}, approximately $1.6\%$ above the theoretical optimum of \textbf{0.549}. To further improve the bound, we increased the grid resolution to $200 \times 200$, with the effect of improving the {\bf Certified Lifted Bound} to \textbf{0.552}, within $0.6\%$ of the theoretical optimum. We display the allocation learned on the coarser $50 \times 50$ grid in \cref{fig:alloc_viz}, which demonstrates the structure of the learned solution, showing
that it resembles the optimal  allocation rule.

These results confirm that our flow-based parameterization is flexible enough to recover complex, non-trivial allocation structures and almost
exact revenue upper bounds.

\subsection{Grid Convergence and Consistency}

We next analyze the convergence behavior of the dual solver in the $2 \times 2$ uniform distribution setting by discretizing the domain into progressively finer grids of size $K \times K$. For small grids, the problem can be solved exactly using a Linear Program (LP), providing a ground truth for dual computational framework.

\Cref{tab:grid_ablation} presents the results. On coarse grids for which the LP is tractable ($K \in \{2, \dots, 10\}$), our {\bf Discrete Dual}  matches the exact LP optimal revenue to within $0.1\%$ and serves as a valid upper bound to the discrete type problem. This validates that our gradient-based optimization successfully navigates the high-dimensional polytope of flow-conserving constraints to find the global maximum.

As the grid resolution $K$ increases, our {\bf Certified Lifted Bound}, which by \cref{Thm:ct} accounts for the continuous geometry within cells to bound the continuous problem, decreases monotonically, providing a progressively stricter upper limit on the optimal revenue (\Cref{tab:grid_ablation}). The gap between the discrete and lifted bounds reduces for finer grids, confirming the consistency of the lifting procedure as we approach a better bound for the continuous problem. The LP solution is intractable for instances with type space larger than $10$ on each dimension,
highlighting the scalability advantage of the neural-network based approach.
\begin{table}[h]
    \caption{Grid refinement analysis (2$\times$2, Uniform). We compare the dual solution obtained by our neural network on the discrete value grid against the exact linear program (LP) solution for the same setting. The learned, discrete bound closely matches the LP optimum, while the lifted bound converges monotonically as the grid refines. Solving the LP becomes intractable for moderate grid sizes: 12$\times$12 and larger.
    \label{tab:grid_ablation}}
    \centering
    \begin{tabular}{r c c c c}
        \toprule
        Grid& \textbf{LP Exact} & \textbf{Discrete Dual (Ours)} & \textbf{Lifted Cont. Bound (Ours)} & Gap (LP vs. Dual) \\
        \midrule
        $2 \times 2$   & 1.5625 & 1.5625 & 1.1861 & $< 0.01\%$ \\
        $4 \times 4$   & 1.2275 & 1.2276 & 1.0422 & $< 0.01\%$ \\
        $8 \times 8$   & 1.0574 & 1.0582 & 0.9597 & $0.07\%$ \\
        $10 \times 10$ & 1.0212 & 1.0234 & 0.9437 & $0.21\%$ \\
        $12 \times 12$ & -- & 1.0000 & 0.9342 & -- \\
        
        $25 \times 25$ & -- & 0.9433 & 0.9117 & -- \\
        $50 \times 50$ & -- & 0.9082 & 0.8919 & -- \\
        \bottomrule
    \end{tabular}
\end{table}

While \cref{lem:bound_uniform} formally relates the optimal discrete revenues of a grid and its refinement, the bounds in \cref{tab:grid_ablation} are produced by our neural network based solver and need not be exactly optimal at each resolution. Nonetheless, the {\bf Certified Lifted Bound} still decreases monotonically as the grid is refined, empirically confirming the consistency of the lifting procedure and indicating that the optimizer tracks the optimal dual closely at every resolution.

\subsection{Multi-Bidder Continuous Value Auctions}

We now present our  results on multi-item, multi-bidder auctions with continuous value distributions, where
GemNet gives candidate optimal mechanisms (and thus a certified lower bound). This is a regime where analytical solutions are unknown, and standard discrete LP formulations are not well-defined because the continuous type space admits no finite variable representation.
We benchmark  the duality framework against  GemNet \citep{wang2024gemnet}, the state of the art deep learning approach that is fully strategyproof (but gives no revenue optimality certificates).
We also benchmark against
RegretNet \citep{dutting2024optimal}, a deep learning approach that optimizes the primal mechanism but with  incentive violations (thus, revenue that may be higher than optimal, but the degree of this possible violation was previously unknown).

\subsubsection{Uniform Valuations: Certified Bounds}
For the Uniform $U[0,1]$ setting, we apply our full pipeline including the lifting map. 
\Cref{tab:multi_results} summarizes the results. For the $2$ bidder, $2$ item setting, we compute a {\bf Certified Upper Bound} of \textbf{0.892} using a $50 \times 50$ grid. 

This bound provides rigorous insights for the two primal methods. It confirms that \textbf{GemNet} (revenue $\approx 0.876$) is performing within \textbf{1.8\%} of the theoretical limit. 
It also reveals that \textbf{RegretNet}, which achieves a revenue of approximately \textbf{0.908},which  exceeds the certified upper bound of $0.892$. This formally establishes that RegretNet's  revenue is not possible in a fully DSIC auction and establishes GemNet as the more reliable benchmark for achievable revenue.

To illustrate these insights visually, we plot the virtual allocation regions induced by the learned dual variables in \Cref{fig:alloc_comparison_methods}. The dual-based approach recovers the sharp, piecewise-linear decision boundaries characteristic of the optimal DSIC mechanism (similar to GemNet), and is distinct from other baselines such as AMenuNet \citep{duan2023scalable}, which limits the design space to affine maximizer allocation rules. In contrast, RegretNet exhibits smoothed boundaries
and comes along with incentive compatibility violations.
\begin{figure*}[t]
    \centering
    \includegraphics[width=\textwidth]{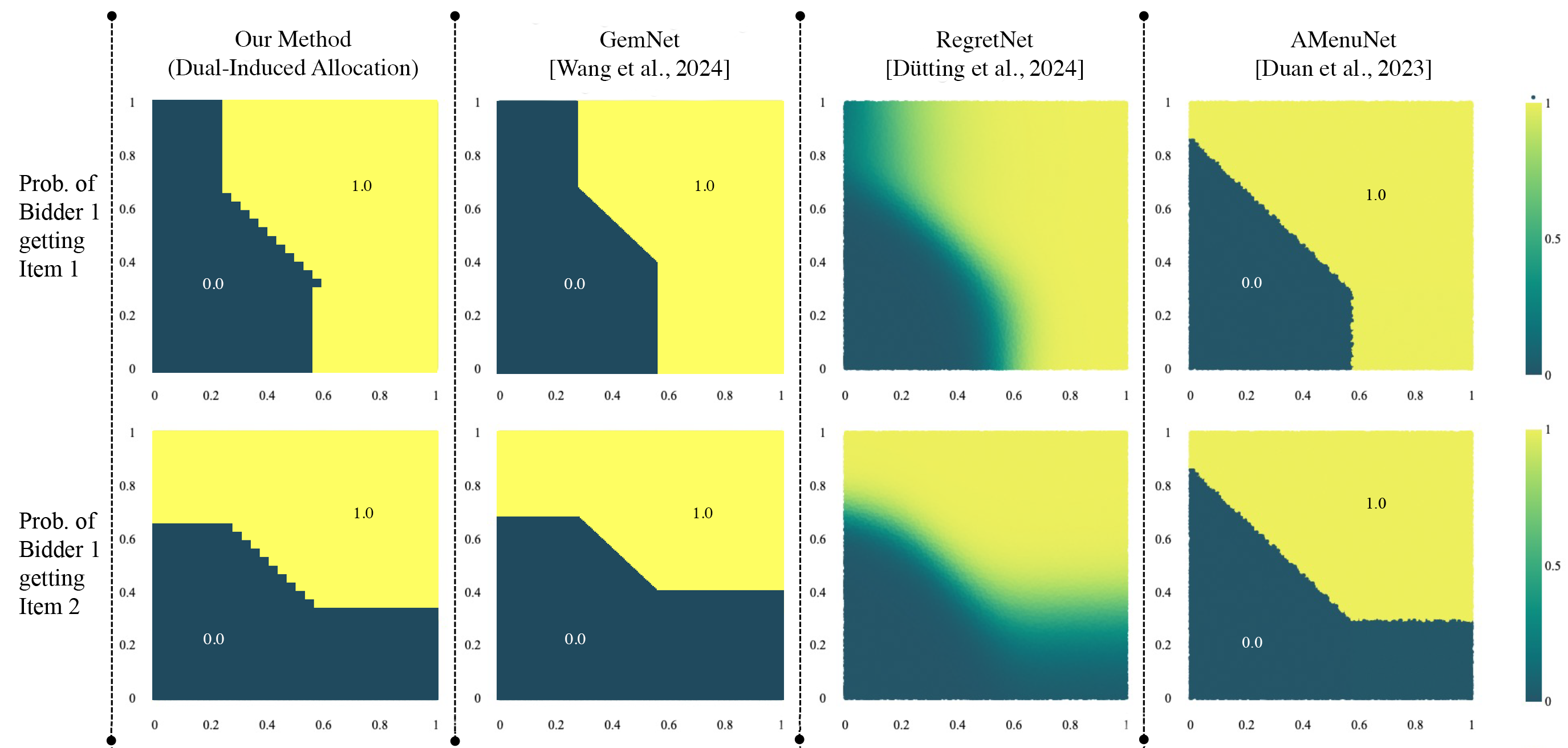}
    \caption{\textbf{Learned dual vs. primal allocations (2$\times$2 Uniform).}  Heatmaps show the probability of Bidder 1 receiving Item 1 (top) and Item 2 (bottom). Bidder 2's values are set at $(0.03, 0.63)$. Our method (left), when using a value grid of 32$\times$32, recovers the sharp decision boundaries characteristic of the known DSIC mechanism (GemNet), whereas RegretNet exhibits smoothed boundaries and has incentive violations and AMenuNet 
    attains a very different allocation rule.
    \label{fig:alloc_comparison_methods}}
\end{figure*}

We also take up  the $3$ bidder, $2$ item problem. We achieve a {\bf Certified Upper Bound} of \textbf{1.120}. Comparing this to the best known primal lower bound of \textbf{1.079} (achieved by GemNet), this certificate bounds the optimality gap of GemNet to within 3.7\%. For both this and the earlier setting, this represents to the best of our knowledge the first rigorous revenue certificate for a DSIC problem of this dimension 
and continuous value distributions.
\begin{table}[t]
\caption{{\bf Certified Upper Bounds} for multi-bidder continuous auctions (2$\times$2, 3$\times$2, Uniform). Our \textbf{Certified Bound} 
reveals that RegretNet overestimates the optimal revenue (due to incentive violations) while  GemNet is within at most 1.8\% of the optimal revenue using the lifted bound from \cref{Thm:ct}.
\label{tab:multi_results}}
\centering
\setlength{\tabcolsep}{5pt} 
\begin{tabular}{l c c c c c}
\toprule
\multirow{2.5}{*}{\textbf{Setting}} & \textbf{GemNet} & \multicolumn{2}{c}{\textbf{RegretNet} (Approx. IC)} & \textbf{Ours} & \textbf{Gap} \\
\cmidrule(lr){3-4}
 & (DSIC) & Revenue & IC Viol. & (Certified Bound) & (vs GemNet) \\
\midrule
2 Bidders, 2 Items & 0.876 & 0.908 & 0.001 & \textbf{0.892} & 1.80\% \\
3 Bidders, 2 Items & 1.079 & 1.111 & 0.002 & \textbf{1.120} & 3.70\% \\
\bottomrule
\end{tabular}
\end{table}

\subsubsection{General Distributions}

To demonstrate the flexibility of our dual computational framework beyond uniform value distributions, we also evaluate an auction design problem for which values are distributed 
Beta(1,2).  Since the lifting map derived in \cref{Thm:ct} 
assumes a uniform value distribution, we report the \textbf{Exact Discrete Bound} in this Beta setting (i.e., the dual objective evaluated exactly over the full discretized type space) and do not use lifting. 
The {\bf Exact Discrete Bound}  is a valid upper bound for the discretized version of
the problem.  As established in \cref{thm:coarse-plus-error}, this discrete bound converges to a valid continuous upper bound as the grid resolution increases.
As shown in \cref{tab:beta_results}, the discrete bound tightly tracks the primal benchmarks. In the $2 \times 2$ setting, the gap between the primal and the discrete dual is approximately 5\%, suggesting that the discretization error is small and the dual approach is effective.
\begin{table}[h]
\caption{Exact Discrete Bound on the discretized grid (2$\times$2, 3$\times$2, Beta(1,2) valuations).
\label{tab:beta_results}}
\centering
\setlength{\tabcolsep}{5pt}
\begin{tabular}{l c c c c c}
\toprule
\multirow{2.5}{*}{\textbf{Setting}} & \textbf{GemNet} & \multicolumn{2}{c}{\textbf{RegretNet} (Approx. IC)} & \textbf{Ours} & \textbf{Gap} \\
\cmidrule(lr){3-4}
 & (DSIC) & Revenue & IC Viol. & (Discrete Bound) & (vs GemNet) \\
\midrule
2 Bidders, 2 Items & 0.558 & 0.578 & 0.004 & \textbf{0.588} & 5.10\% \\
3 Bidders, 2 Items & 0.733 & 0.734 & 0.003 & \textbf{0.760} & 3.55\% \\
\bottomrule
\end{tabular}
\end{table}

\section{Closing Remarks}

In this work we present a deep-learning based framework for computing certified  upper bounds 
on the  revenue in optimal, multi-item, multi-bidder DSIC auctions. By parameterizing the dual problem via absorbing Markov chains, we enforce a flow conservation property---necessary for a useful dual solution---structurally, transforming a complex constrained optimization problem into a simpler, unconstrained optimization problem. Our computational framework and theoretical results bridge the gap between discrete computational methods and optimal auction design for continuous value distributions, providing the first rigorous computational framework to obtain certificates of near-optimality in the uniform distribution setting for state-of-the-art mechanisms identified through methods such as GemNet, while showing that approximately incentive-compatible baselines such as RegretNet overestimate achievable revenue.

While our approach scales substantially beyond linear programming, which becomes intractable for the discretized problems studied here, it remains subject to the curse of dimensionality inherent to grid-based discretizations. A key advantage of the neural network parameterization, however, is that optimization is performed over a continuous, low-dimensional parameter space using stochastic gradient descent on sampled profiles, even though the dual objective is ultimately evaluated over the enumerated discrete type space. Extending the framework to avoid explicit grid enumeration entirely, for example through sparse sampling over the type space, is an important direction for future work.

\newpage
\bibliographystyle{ACM-Reference-Format}
\bibliography{sample-bibliography}

\newpage
\appendix
\crefalias{section}{appendix}
\crefname{appendix}{Appendix}{Appendices}
\Crefname{appendix}{Appendix}{Appendices}

\newpage

\section{Proofs}\label{appx:proof}
\flowconservation*
\begin{proof}
     If flow conservation fails for some $i$ and $v_{-i}$, then there exists a $v_i\in \mathrm{supp}(D^\mathcal{G}_i)$ for which $L(\lambda,x,p)$ becomes unbounded as a function of $p_i(v)$. 
\end{proof}

\vw*
\begin{proof}
The proof proceeds by using flow conservation to simplify $L(\lambda,x,p)$ into a virtual-welfare expression and observing that the Lagrange-multiplier terms are nonnegative over $P_+(D^\mathcal{G})$, becoming zero exactly when the corresponding constraints bind; strong duality then yields the stated characterization for $(x^*,p^*)$ and $\lambda^*$. 
\end{proof}

\lcf*

\begin{proof}
Fix bidder $i$. Fix almost every $v_{-i}$ and set $t_{-i}=R(v_{-i})$.
Fix $v_i$ and set $t_i=R(v_i)$.

\paragraph{Step 1: compute total outflow from $v_i$.}
By \cref{eq:lift-kernel} and \cref{eq:lift-kernel-sink},
\begin{align*}
\lambda_i^{v_{-i}}(v_i,V_i\cup\{\varnothing\})
&=
\lambda_i^{v_{-i}}(v_i,V_i)+\lambda_i^{v_{-i}}(v_i,\{\varnothing\})\\
&=
s_{-i}(v_{-i})\,s_i(v_i)\,
\sum_{t_i'\in\overline{T}_i}
\lambda^{\mathcal{G},t_{-i}}_i(t_i,t_i')\,
\underbrace{\int_{C_i(t_i')} s_i(u_i)\,\mathrm{d}u_i}_{=\,1}
\;+\;
s_{-i}(v_{-i})\,s_i(v_i)\,\lambda^{\mathcal{G},t_{-i}}_i(t_i,\varnothing)\\
&=
s_{-i}(v_{-i})\,s_i(v_i)\,
\sum_{t_i'\in\overline{T}_i\cup\{\varnothing\}}
\lambda^{\mathcal{G},t_{-i}}_i(t_i,t_i').
\end{align*}

\paragraph{Step 2: use discrete flow conservation.}
\[
\sum_{t_i'\in\overline{T}_i\cup\{\varnothing\}}
\lambda^{\mathcal{G},t_{-i}}_i(t_i,t_i')
=
f^\mathcal{G}(t_i,t_{-i})
+
\sum_{t_i'\in\overline{T}_i}\lambda^{\mathcal{G},t_{-i}}_i(t_i',t_i).
\]
Hence
\begin{equation}\label{eq:outflow-expanded}
\lambda_i^{v_{-i}}(v_i,V_i\cup\{\varnothing\})
=
s_{-i}(v_{-i})\,s_i(v_i)\,f^\mathcal{G}(t_i,t_{-i})
+
s_{-i}(v_{-i})\,s_i(v_i)\,\sum_{t_i'\in\overline{T}_i}\lambda^{\mathcal{G},t_{-i}}_i(t_i',t_i).
\end{equation}

\paragraph{Step 3: identify the source term.}
Because $f^\mathcal{G}(t_i,t_{-i})=f^\mathcal{G}_i(t_i)f^\mathcal{G}_{-i}(t_{-i})$,
\[
s_{-i}(v_{-i})\,s_i(v_i)\,f^\mathcal{G}(t_i,t_{-i})
=
\frac{f_{-i}(v_{-i})}{f^\mathcal{G}_{-i}(t_{-i})}\cdot\frac{f_i(v_i)}{f^\mathcal{G}_i(t_i)}\cdot
f^\mathcal{G}_i(t_i)f^\mathcal{G}_{-i}(t_{-i})
=
f_i(v_i)f_{-i}(v_{-i})
=
f(v_i,v_{-i}).
\]
Thus the first term on the right-hand side of \cref{eq:outflow-expanded} is
exactly the source inflow $f(v)$ required by \cref{eq:flow-cont}.

\paragraph{Step 4: compute total inflow into $v_i$.}
We claim
\begin{equation}\label{eq:inflow}
\int_{V_i}\lambda_i^{v_{-i}}(u_i,\mathrm{d}v_i)
=
s_{-i}(v_{-i})\,s_i(v_i)\,\sum_{t_i'\in\overline{T}_i}\lambda^{\mathcal{G},t_{-i}}_i(t_i',t_i).
\end{equation}
To see this, note that by \cref{eq:lift-kernel}, the measure
$\lambda_i^{v_{-i}}(u_i,\cdot)$ restricted to the cell $C_i(t_i)$ has density
\[
\text{(w.r.t.\ Lebesgue on $C_i(t_i)$)}\qquad
s_{-i}(v_{-i})\,s_i(u_i)\,\lambda^{\mathcal{G},t_{-i}}_i(R(u_i),t_i)\,s_i(v_i),
\]
because $\int_{A\cap C_i(t_i)} s_i(w)\,\mathrm{d}w$ contributes the factor
$s_i(v_i)$ at the point $v_i$.
Integrating over all $u_i$ and summing over the finitely many source-cells for
$u_i$ gives
\begin{align*}
\int_{V_i}\lambda_i^{v_{-i}}(u_i,\mathrm{d}v_i)
&=
\sum_{t_i'\in\overline{T}_i}\int_{C_i(t_i')} s_{-i}(v_{-i})\,s_i(u_i)\,
\lambda^{\mathcal{G},t_{-i}}_i(t_i',t_i)\,s_i(v_i)\,\mathrm{d}u_i\\
&=
s_{-i}(v_{-i})\,s_i(v_i)\,\sum_{t_i'\in\overline{T}_i}\lambda^{\mathcal{G},t_{-i}}_i(t_i',t_i)
\underbrace{\int_{C_i(t_i')} s_i(u_i)\,\mathrm{d}u_i}_{=\,1},
\end{align*}
which is exactly \cref{eq:inflow}.

\paragraph{Step 5: conclude flow conservation.}
Combine \cref{eq:outflow-expanded} with the identifications in Steps 3 and 4:
\[
\lambda_i^{v_{-i}}(v_i,V_i\cup\{\varnothing\})
=
f(v_i,v_{-i})
+
\int_{V_i}\lambda_i^{v_{-i}}(u_i,\mathrm{d}v_i),
\]
which is precisely \cref{eq:flow-cont} for almost every $(v_i,v_{-i})$.
\end{proof}

\cwd*

\begin{proof}

The proof is in \cref{sec:method:continuous}, and here we address the final step.

From \cref{eq:cont-thm24}, for the \emph{fixed} flow-conserving kernel $\lambda$ (hence fixed $\Phi^\lambda$),
we have for every DSIC/IR feasible mechanism $(x,p)$:
\[
\mathrm{Rev}_{(x,p)}(D)
=\int_V f(v)\Big(\sum_{i=1}^n p_i(v)\Big)\,dv
\;\le\;
\int_V f(v)\Big(\sum_{i=1}^n\sum_{j=1}^m x_{ij}(v)\,\Phi^\lambda_{ij}(v)\Big)\,dv.
\]
Now take the supremum over all DSIC/IR feasible mechanisms on the left-hand side. Since the
right-hand side depends on $(x,p)$ only through the allocation rule $x$, we obtain
\[
\mathrm{Rev}(D)
:=\sup_{\substack{(x,p)\ \text{DSIC, IR, feasible}}}\ \mathrm{Rev}_{(x,p)}(D)
\;\le\;
\sup_{\substack{(x,p)\ \text{DSIC, IR, feasible}}}
\int_V f(v)\Big(\sum_{i,j} x_{ij}(v)\,\Phi^\lambda_{ij}(v)\Big)\,dv.
\]

\end{proof}

\ubu*

\begin{proof}
Let the (coarse) grid be denoted by $\mathcal{G}^C$ with mesh size $\Delta$, and let $D^C$ be the
induced discretized distribution on $\mathcal{G}^C$.  Let $\lambda^C$ be any \emph{flow-conserving} dual
solution for the discrete problem on $\mathcal{G}^C$ (hence it induces coarse virtual values
$\Phi^{\lambda^C}_{ij}(w)$ for $w\in\mathcal{G}^C$).  Fix an integer $K\ge 2$ and consider the $K$-refined
grid $\mathcal{G}^F$ (mesh $\Delta/K$) together with its discretized distribution $D^F$.

By \cref{lem:bound_uniform}, the optimal revenue under the refined discretization
is upper bounded by:
\[
\mathrm{Rev}(D^F)
\;\le\;
\mathbb{E}_{w\sim D^C}\,\mathbb{E}_{v\sim S(w)}
\Bigg[\sum_{j=1}^m \max\Bigl\{0,\max_{i\in[n]}
\bigl(\Phi^{\lambda^C}_{ij}(w) - (w_{ij}-v_{ij})\bigr)\Bigr\}\Bigg],
\]
where $S(w)\subset \mathcal{G}^F$ denotes the (finite) set of refined grid points whose parent is $w$
under the parent map $R(\cdot)$ (so $R(v)=w$ for $v\in S(w)$). 

Let $(\mathcal{G}^{(r)})_{r\ge 1}$ be any sequence of refinements of $\mathcal{G}^C$ whose mesh sizes
tend to $0$ (equivalently, take $K\to\infty$). In the uniform case, conditioning on a parent
$w\in\mathcal{G}^C$, the discrete conditional distribution on the set of children $S_r(w)\subset\mathcal{G}^{(r)}$
converges to the uniform distribution over the cell
\[
C(w)\;:=\;\{v\in V : R(v)=w\},
\]
and the offset term $(w_{ij}-v_{ij})$ becomes the within-cell slack of the continuous type $v$ below its
parent $w$. Therefore the right-hand side above converges to
\[
\mathbb{E}_{w\sim D^C}\,\mathbb{E}_{v\sim C(w)}
\Bigg[\sum_{j=1}^m \max\Bigl\{0,\max_{i\in[n]}
\bigl(\Phi^{\lambda^C}_{ij}(w) - (w_{ij}-v_{ij})\bigr)\Bigr\}\Bigg].
\]
Here $v\sim C(w)$ denotes the conditional distribution of $v$ given $R(v)=w$, which is uniform over
the cell under the uniform distribution.

\cref{sec:method:continuous} lifts $\lambda^C$ to a continuous kernel $\lambda$ by spreading the discrete flow within each cell;
\cref{lem:lift-flow} guarantees that this lift is flow-conserving. 
Then \cref{lem:cont-thm24} (continuous weak duality) implies that any flow-conserving $\lambda$ yields an upper bound
on the optimal continuous DSIC revenue:
\[
\mathrm{Rev}(D)\;\le\;\int_V f(v)\Big(\sum_{j=1}^m \max_{i\in[n]}(\Phi^\lambda_{ij}(v))_+\Big)\,dv.
\]
Moreover, under the uniform lifting construction, the induced continuous virtual values satisfy the
explicit within-cell relation
\(
\Phi^\lambda_{ij}(v)=\Phi^{\lambda^C}_{ij}(R(v))-(R(v)_{ij}-v_{ij}),
\)
which is the continuous analogue of the refined-grid formula. Plugging this identity into the \cref{lem:cont-thm24} bound and rewriting the integral by conditioning on $w=R(v)$ yields exactly
\[
\mathrm{Rev}(D)
\;\le\;
\mathbb{E}_{w\sim D^C}\,\mathbb{E}_{v\sim C(w)}
\Bigg[\sum_{j=1}^m \max\Bigl\{0,\max_{i\in[n]}
\bigl(\Phi^{\lambda^C}_{ij}(w) - (w_{ij}-v_{ij})\bigr)\Bigr\}\Bigg].
\]
\end{proof}

\section{Implementation Details}
\label{app:implementation}

\subsection{Computing Infrastructure and Runtime}

All experiments were implemented in PyTorch and executed on a cluster equipped with NVIDIA A100 (80GB) GPUs. Small-scale experiments (e.g., 2$\times$2) were trained on a single GPU, while larger instances (e.g., 3$\times$2) utilized Distributed Data Parallel (DDP) across up to 4 GPUs. Typical training times ranged from 5 to 10 hours for the largest instance (35$\times$35 grid, 3 bidders and 2 items) to a few seconds for the smallest (Yao auction) setting, depending on the grid resolution and number of bidders.

\subsection{Network Architecture}

We employ two distinct architectural variants depending on the complexity of the type space.

\paragraph{Standard MLP (Discrete and Low-Dimensional Settings).}

For discrete domains like the Yao auction or coarse discretizations of small continuous settings, we parameterize the flow policy using a standard Multi-Layer Perceptron (MLP).

\begin{itemize}
\item \textbf{Input:} The valuation profile of opposing bidders $v_{-i}$ (concatenated with a one-hot bidder ID for symmetric settings).
\item \textbf{Architecture:} A 3-layer MLP with ReLU or GELU activations and a hidden dimension of 256 or 512.
\item \textbf{Output:} The network directly outputs a flattened vector of logits corresponding to the transition matrix $\Pi \in \mathbb{R}^{K \times K}$ and sink vector $\alpha \in \mathbb{R}^K$.
\end{itemize}
This architecture provides a strong baseline and successfully recovers optimal revenues for known analytical cases (e.g., Table 1 in main text).

\paragraph{Fourier-Bilinear Networks (Scaled Continuous Settings).}

To scale to high-resolution grids (e.g., $50 \times 50$ per bidder) and larger numbers of bidders, explicitly outputting a $K^m \times K^m$ transition matrix becomes computationally prohibitive. For these regimes, we introduce a factorized architecture:
\begin{enumerate}
\item \textbf{Fourier Encoding:} We map raw grid coordinates $v \in [0,1]^m$ to high-dimensional features using Fourier embeddings $\gamma(v) = [\sin(2^k \pi v), \cos(2^k \pi v)]_{k=0}^{S-1}$ with $S=10$ scales. This allows the network to capture high-frequency variations in the optimal dual variables.
\item \textbf{Bilinear Factorization:} Instead of predicting the full transition matrix $\Pi_{uv}$, we predict low-rank factors. A \emph{Source Network} maps context $v_{-i}$ and source type $u$ to a query vector $q_u$, while a \emph{Target Network} maps target type $w$ to a key vector $k_w$. The flow logits are computed as the inner product $L_{uw} = \langle q_u, k_w \rangle$.
\end{enumerate}
This factorization decouples the parameter count from the square of the grid size, enabling efficient training on fine-grained discretizations.

The Fourier-Bilinear architecture
is only used for the multi-bidder, multi-item continuous
experiments at grid resolution $K=50$ (the $2\times 2$ and $3\times 2$ uniform
results in \cref{tab:multi_results} and the $2\times 2$ and $3\times 2$ Beta
results in \cref{tab:beta_results}); all other experiments use the standard
MLP variant.

\subsection{Differentiable Flow Layer}
The output of the neural network (whether MLP or Bilinear) is treated as the raw logits of an absorbing Markov chain. To strictly enforce the dual feasibility constraints, we apply a softmax function over the logits to obtain transition probabilities. Crucially, we clamp the sink probability $\alpha$ to be at least $\epsilon > 0$, ensuring the transition matrix $\Pi$ is strictly sub-stochastic.

The steady-state flow $\mu$ is computed by solving the linear system $(I - \Pi^\top)^{-1}\nu = \mu$. This operation is fully differentiable, allowing us to backpropagate the revenue objective through the equilibrium flow.

\subsection{Training Protocol}

For the optimization hyperparameters, we utilize the AdamW optimizer coupled with a cosine annealing scheduler (including a warm-up phase). The learning rate is initialized at a maximum of $3 \times 10^{-4}$ for continuous settings (or $1 \times 10^{-3}$ for discrete settings) and decays to a minimum of $1 \times 10^{-6}$. To ensure the network explores diverse flows before tightening the feasibility constraints, the sink probability parameter $\epsilon$ is annealed from a starting value of $1 \times 10^{-2}$ down to $5 \times 10^{-6}$ over the course of 150,000 to 160,000 total training steps.

\subsection{Evaluation and Certification} 

\paragraph{Discrete Validation.} For discrete settings where the type space is finite (e.g., the auction studied by Yao with binary valuation space), the dual objective is computed exactly by summing over the full support. This serves as a validation step, confirming that our neural architecture can recover known analytical solutions without discretization error.

\paragraph{Uniform Distribution.} For uniform valuations, we compute the \textbf{Certified Lifted Bound} by lifting the discrete dual variables to the continuous domain and integrating the resulting piecewise-linear functions exactly, as detailed in Section~\ref{sec:method:lift}.

\paragraph{Beta Distribution.} To maintain consistency with the uniform grid formulation established in our theoretical framework, we discretize the Beta(1,2) distribution using a uniform value grid. The probability mass for each cell is computed exactly from the cumulative distribution function (CDF).
However, the dual variables involve terms inversely proportional to the probability density ($1/f(v)$). For distributions where the density vanishes at the boundary (e.g., $f(1) = 0$ for Beta(1,2)), this creates numerical singularities. To resolve this while preserving the uniform grid structure, we employ a numerical clamping strategy: we restrict the effective computation grid to the interval $[0, 0.99]$. This ensures numerical stability for the inverse density terms while maintaining the uniform discretization structure required by our convergence certificates. The reported bounds are the Exact Discrete Bounds computed on this clamped uniform grid.

\end{document}